\documentclass[11pt, a4paper]{article}

\usepackage[margin=1in]{geometry}
\usepackage{amsmath, amssymb, amsthm}
\usepackage{hyperref}
\usepackage{graphicx}
\usepackage{bbm}

\hypersetup{
    colorlinks=true,
    linkcolor=blue,
    filecolor=magenta,      
    urlcolor=cyan,
}

\newtheorem{theorem}{Theorem}[section]
\newtheorem{proposition}[theorem]{Proposition}

\newtheorem{corollary}[theorem]{Corollary}
\theoremstyle{definition}

\theoremstyle{remark}
\newtheorem{remark}[theorem]{Remark}

\numberwithin{equation}{section}
\numberwithin{theorem}{section}

\title{\bf Edge density expansions for the classical Gaussian and Laguerre ensembles}
\author{}
\date{}

\author{Peter J. Forrester${}^1$,
Anas A. Rahman${}^2$, and 
Bo-Jian Shen${}^1$}
\date{}

\begin{document}

\maketitle

${}^1$School of Mathematics and Statistics,  The University of Melbourne,
Victoria 3010, Australia. \: \: Email: {\tt pjforr@unimelb.edu.au};  {\tt bojian.shen@unimelb.edu.au}\\

${}^2$
Department of Mathematics,
The University of Hong Kong, Hong Kong. \\
Email: {\tt aarahman@hku.hk}

\bigskip

\

\begin{abstract}
Recent work of Bornemann has uncovered hitherto hidden integrable structures relating to the asymptotic expansion of quantities at the soft edge of Gaussian and Laguerre random matrix ensembles. These quantities are spacing distributions and the eigenvalue density, and the findings cover the cases of the three symmetry classes orthogonal, unitary and symplectic. In this work we give a different viewpoint on these results in the case of the soft edge scaled density, and in the Laguerre case we initiate an analogous study at the hard edge.
Our tool is the scalar differential equation satisfied by the latter, known from earlier work. Unlike integral representations, these differential equations in soft edge scaling variables isolate the function of $N$ which is the expansion variable. Moreover, they give information on the correction terms which supplements the findings from the work of Bornemann. In the case of the Gaussian ensemble, we can demonstrate analogous features for Dyson index $\beta = 6$, which suggests a broader class of models, namely the classical
$\beta$ ensembles, with asymptotic expansions exhibiting integrable features. For the Laguerre ensembles at the hard edge, we give the explicit form of the correction at second order for unitary symmetry, and at first order in the orthogonal and symplectic cases. Various differential relations are demonstrated.
\end{abstract}


\section{Introduction}
\subsection{The Gaussian and Laguerre random matrix ensembles}
For $G$ an $N \times N$ real or complex standard Gaussian matrix, let us form the Hermitian matrix $H = {1 \over 2}(G+ G^\dagger)$. In the real case this is the construction of random matrices forming the Gaussian orthogonal ensemble (GOE), while in the complex case it specifies matrices from the 
Gaussian unitary ensemble (GUE). An extensive account relating to these ensembles is given in \cite[Ch.~1]{Fo10}. One key property is that their eigenvalue probability density functions have the functional form proportional to \cite[Prop.~1.3.4]{Fo10}
\begin{equation}\label{1.1}
\prod_{l=1}^N e^{-\beta x_l^2/2} \prod_{1 \le j < k \le N}
|x_k - x_j|^\beta.
\end{equation}
Here, the parameter $\beta$ --- known as the Dyson index --- takes the values 1 and 2 for the GOE and GUE respectively.

The above additive construction of non-Hermitian random matrices to give Hermitian random matrices has a multiplicative counterpart. Now, for
$n \ge N$, let $X$ be a $n \times N$ rectangular real (complex) standard Gaussian matrix (also referred to as rectangular real and complex Ginibre matrices \cite{BF25}), and form the positive definite random matrix $W = X^\dagger X$. This construction specifies the real (complex) Wishart ensemble; see e.g.~\cite[\S 3.2]{Fo10}.
Replacing (\ref{1.1}) for the corresponding eigenvalue probability density functions is the functional form proportional to
\cite[Prop.~3.2.2]{Fo10}
\begin{equation}\label{1.2}
\prod_{l=1}^N x_l^a e^{-\beta x_l/2} \mathbbm 1_{x_l > 0}\prod_{1 \le j < k \le N}
|x_k - x_j|^\beta, \quad a:= {\beta \over 2} (n - N + 1) - 1,
\end{equation}
with $\beta = 1$ ($\beta = 2$) for the real (complex) Wishart ensemble. Allowing the parameter $a$ in (\ref{1.2}) to take on continuous values $a > -1$, this is said to specify the Laguerre orthogonal ensemble (LOE) for $\beta = 1$, and the Laguerre unitary ensemble (LUE) for $\beta = 2$.

The constructions leading to (\ref{1.1}) and (\ref{1.2}) can be extended to give realisations of the case $\beta = 4$. This comes about by choosing $G$ and $X$ to be $N \times N$ and $n \times N$ matrices with the elements $2 \times 2$ complex matrices representing quaternions; see e.g.~\cite[Eq.~(1.20)]{Fo10}. Choosing the two independent entries in the latter to be standard complex Gaussians and forming $H$ and $W$ as in the above paragraphs gives Hermitian random matrices with a doubly degenerate spectrum. The probability density function of the $N$ independent eigenvalues in the spectrum is given by (\ref{1.1}) for $H$, and (\ref{1.2}) for $W$.

\subsection{Outline of some findings of Bornemann}\label{S1.2}
Up to proportionality, the eigenvalue probability density functions (\ref{1.1}) and (\ref{1.2}) are of the form
\begin{equation}\label{1.3}
\prod_{l=1}^N w(x_l) \prod_{1 \le j < k \le N}
|x_k - x_j|^\beta,
\end{equation}
for weight function $w(x)$ given by $w^{\rm G}(x):= e^{-\beta x^2/2}$ and
$w^{\rm L}(x) := x^a e^{-\beta x/2} \mathbbm 1_{x > 0}$ respectively. Let us denote by $E_{N,\beta}(k;(s,\infty);w(x))$ the probability that in an ensemble with eigenvalue probability density function proportional to (\ref{1.3}), there are $k$ eigenvalues in the semi-infinite interval $(s,\infty)$. Introduce too the generating function
\begin{equation}\label{1.4}
\mathcal E_{N,\beta}((s,\infty);w(x);\xi) := \sum_{k=0}^\infty (1 - \xi)^k
E_{N,\beta}(k;(s,\infty);w(x)).
\end{equation}
A fundamental result in random matrix theory is that by centering the variable $s$ to the neighbourhood of the largest eigenvalue, and choosing a scale so that in that neighbourhood the mean eigenvalue spacings are of order unity, the generating function (\ref{1.4}) has a well defined limiting form. Moreover, for a large class of weights, this limiting form is independent of the weight $w(x)$ (see \cite{DG07,BEY14,KRV16}), an effect referred to as universality. 

Explicit computation of the limiting form is simplest in the case $\beta = 2$. Then, the eigenvalues form a determinantal point process (see, e.g., \cite[Ch.~5]{Fo10}), meaning that the $k$-point correlation
$\rho_{(k),N}(x_1,\dots,x_k)$ is specified by a particular correlation kernel 
$K_N(x,y)$ according to
\begin{equation}\label{1.5}
\rho_{(k),N}(x_1,\dots,x_k) = \det [ K_N(x_i,x_j)]_{i,j=1,\dots,k}.
\end{equation}
As a consequence (see, e.g., \cite[\S 9.1]{Fo10}), 
\begin{equation}\label{1.6}
\mathcal E_{N,\beta}((s,\infty);w(x);\xi) \Big |_{\beta = 2} =
\det (\mathbb I - \xi \mathbb K_{(s,\infty)}),
\end{equation}
where $\mathbb K_{(s,\infty)}$ is the integral operator on $(s,\infty)$
with kernel $K_N(x,y)$. In the case of the GUE, the latter is given in terms of Hermite polynomials (see \cite[\S 5.4]{Fo10}). From this starting point, it was shown in \cite{Fo93} that
\begin{equation}\label{1.7}
 \mathcal E_{\beta}^{\rm s}((y,\infty);\xi) \Big |_{\beta = 2} :=
\lim_{N \to \infty} \mathcal E_{N,\beta}((\sqrt{2N} + y/(\sqrt{2} N^{1/6}),\infty);e^{-x^2};\xi) \Big |_{\beta = 2} =
\det ( \mathbb I - \xi \mathbb K^{\rm s}_{(y,\infty)}).
\end{equation}
Here, $ \mathbb K^{\rm s}_{(y,\infty)}$ is the integral operator on $(y,\infty)$ with kernel 
\begin{equation}\label{1.8}
K^{\rm s}(x,y) = {{\rm Ai}(x) {\rm Ai}'(y) - {\rm Ai}(y) {\rm Ai}'(x) \over x - y},
\end{equation}
now known as the Airy kernel. The superscript label ``s'' indicates that soft edge scaling has been applied: Comparing with (\ref{1.6}), we see that the end point of the semi-infinite interval has been chosen to equal $\sqrt{2N} + y/(\sqrt{2} N^{1/6})$. This is in keeping with the position of the largest eigenvalue in the GUE being to leading order located at $\sqrt{2N}$, and with $1/N^{1/6}$ being the scale of the spacing between eigenvalues in this neighbourhood. Note, though, that the limiting functional form $ \mathcal E_{\beta}^{\rm s}((y,\infty);\xi) |_{\beta = 2}$ is valid for a much larger class of models than that of the GUE, in accordance with the universality phenomenon discussed in the previous paragraph.

There is another prominent form of the limiting soft edge generating function $\mathcal E_{N,\beta}$ in the case $\beta = 2$
\cite{TW94}. For this, denote by
$q(x;\xi)$ the solution of the particular Painlev\'e II equation
with prescribed limiting behaviour
\begin{equation}\label{1.9} 
q'' = x q + 2 q^3, \quad q(x;\xi) \mathop{\sim}\limits_{x \to \infty}
\sqrt{\xi} {\rm Ai}(x).
\end{equation}
Thus, it was shown in \cite{TW94} that
\begin{equation}\label{1.10}
\det ( \mathbb I - \xi \mathbb K^{\rm s}_{(y,\infty)}) =
\exp \bigg ( - \int_y^\infty (x-y) (q(x;\xi))^2 \, dx \bigg ).
\end{equation}
This form makes it particularly clear that, from a mathematical viewpoint, $ \mathcal E_{\beta}^{\rm s}((y,\infty);\xi) |_{\beta = 2}$ relates to  integrable systems.

Thirty or so years after the results of \cite{TW94}, in a series of recent works \cite{Bo24,Bo25,Bo25a}, Bornemann has uncovered integrability properties associated with the large $N$ expansion of the 
$N$-dependent soft edge scaled generating function in the first equality of
(\ref{1.7}). Thus, it has been shown that this quantity admits the large $N$ asymptotic expansion, in powers of $N^{-2/3}$,
\begin{multline}\label{1.11}
\mathcal E_{N,2}((\sqrt{2N} + y/(\sqrt{2} N^{1/6}),\infty);e^{-x^2};\xi)  \\ = 
 \mathcal E_{2}^{\rm s}((y,\infty);\xi)  +
 {1 \over 4 N^{2/3}}  \mathcal E_{2}^{1,\rm s, G}((y,\infty);\xi)  + {1 \over 16 N^{4/3}}  \mathcal E_{2}^{2,\rm s, G}((y,\infty);\xi)  + \cdots,
 \end{multline}
 with
 \begin{align}\label{1.12}
 \begin{aligned}
  & \mathcal E_{2}^{1,\rm s, G}((y,\infty);\xi) =  \Big ( {y^2 \over 5} {d \over dy} - {3 \over 10} {d^2 \over d y^2} \Big )  \mathcal E_{2}^{\rm s}((y,\infty);\xi), \\
   & \mathcal E_{2}^{2,\rm s, G}((y,\infty);\xi) =  \bigg ( 
  -\Big ({141  \over 350}  + {8 y^3 \over 175} \Big ) {d \over dy} +
  \Big ( {39 y \over 175} + {y^4 \over 50} \Big ) {d^2 \over d y^2}
   \\
 & \hspace*{6cm}- {3 y^2 \over 50}{d^3 \over d y^3} 
  + {9 \over 200} {d^4 \over d y^4} \bigg )  \mathcal E_{2}^{\rm s}((y,\infty);\xi).
  \end{aligned}
  \end{align}
  In the first and second corrections, the superscript G indicates the Gaussian case. Unlike the leading term, these functional forms are not universal. Nonetheless, there is a weak form of universality in the sense that in the LUE case, the soft edge expansion is again in powers of $N^{-2/3}$.
 The expansion (\ref{1.11}) continues, with the explicit form of 
 $\mathcal E_{2}^{3,\rm s, G}((y,\infty);\xi)$, the coefficient of $N^{-2}/64$, given in \cite[Th.~2.1]{Bo24} with $F(t)$ therein replaced by $ \mathcal E_{2}^{\rm s}((t,\infty);\xi)$. 
 
 One should be aware, for a start, that the existence of an expansion in powers of $N^{-2/3}$ is far from obvious when beginning with the orthogonal polynomial expression of the correlation kernel $K_N(x,y)$ (we know from, e.g., \cite[Lemma 2.1]{BFM17} that an expansion of the kernel can be lifted to an expansion of the Fredholm determinant). In \cite{Bo24}, this is checked to high order by establishing the corresponding soft edge expansion of the Hermite polynomials to high order; it is expected that a general proof using a Riemann--Hilbert analysis can be given analogous to that provided in \cite{YZ23}, in a related setting.
 Less obvious still is that the successive correction terms should be given by $\xi$-independent, polynomial-weighted differential operators applied to the limiting distribution. This is readily appreciated if one recalls from \cite{FT18}
 the complexity of functional forms of $ \mathcal E_{2}^{1,\rm s, G}((y,\infty);\xi)$ available from earlier literature.

 In \cite{Bo25}, it is noted that the structures (\ref{1.11}) and (\ref{1.12}) have consequences  for the soft edge  expansion of the density
 $\rho_{(1),N}^{\rm GUE}(x)$. This comes about from the fact that (see, e.g., \cite[Eq.~(9.1)]{Fo10}) generally,
 \begin{equation}\label{1.13}
 \mathcal E_{N,\beta}((s,\infty);w(x);\xi) = 1 + \xi \int_s^\infty 
 \rho_{(1),N}(x) \, dx + \dots,
 \end{equation}
 where the terms not shown on the right-hand side are of higher powers in $\xi$ (in distinction to (\ref{1.11}), the expansion (\ref{1.13}) is convergent for all $|\xi| < 1$). Here, the dependence of $ \rho_{(1),N}(x)$ on $w(x)$ and $\beta$ in its notation have been suppressed for notational clarity. Let us now specialise to the GUE and expand the soft edge scaled density
  \begin{equation}\label{1.14}
  {1 \over \sqrt{2} N^{1/6}} \rho_{(1),N}^{\rm GUE}(\sqrt{2N} + y/(\sqrt{2}
  N^{1/6})) = \rho_{(1)}^{\rm s}(y) + {1 \over 4 N^{2/3}}
  \rho_{(1)}^{1,\rm s,G}(y) + {1 \over 16 N^{4/3}}
  \rho_{(1)}^{2,\rm s,G}(y) + \cdots ,
  \end{equation} 
as is consistent with (\ref{1.11}). On the right-hand side, we have suppressed the fact that we are specialising to $\beta = 2$, although in the first and second corrections we have included a superscript G for the same reason as in (\ref{1.11}). We recall from \cite{Fo93} the explicit (universal for $\beta = 2$) functional form of the limiting soft edge density,
 \begin{equation}\label{1.15}
  \rho_{(1)}^{\rm s}(y) =  ( {\rm Ai}'(y)  )^2 - y  ( {\rm Ai}(y)  )^2.
\end{equation}  
 With the differential operators in (\ref{1.12}) denoted
 $\mathcal D^{1,\rm GUE}$, $\mathcal D^{2,\rm GUE}$ respectively (and similarly at higher order although not made explicit in our presentation), it follows by taking the derivative in $\xi$ at $\xi=0$ on both sides of (\ref{1.11}), as prescribed by (\ref{1.13}), that for $j=1,2,\dots$,
  \begin{equation}\label{1.16}
  \int_y^\infty  \rho_{(1)}^{j,\rm s,G}(x) \, dx =
  \mathcal D^{ j,\rm GUE}  \int_y^\infty  \rho_{(1)}^{\rm s}(x) \, dx.
  \end{equation}  
  Moreover, noting from (\ref{1.12}) the factorisation
  $ \mathcal D^{j,\rm GUE} = \tilde{D}^{j,\rm GUE} {d \over d y} $,
  with $\tilde{D}^{j,\rm GUE}$ a polynomial differential operator, valid
  for $j=1,2$ at least, it follows by differentiating both sides of  (\ref{1.16}) that
  \begin{equation}\label{1.17}
    \rho_{(1)}^{j,\rm s,G}(y) = {d \over d y}  \tilde{D}^{j,\rm GUE}
   \rho_{(1)}^{\rm s}(y) .
 \end{equation}  
That is, the correction terms of (\ref{1.14}) are given by polynomial-weighted differential operators applied to the limiting distribution.

Most important for our purposes are two further aspects of 
Bornemann's findings. The first is that, with an appropriate shifting of $N$ by way of the variable
$N_{\rm s}' = N + (\beta - 2)/\beta$ (this is equation 
(\ref{Nd}) below),
a soft edge expansion in  powers of $(N_{\rm s}')^{-2/3}$ holds for the GOE and GSE density and, more generally, the generating function
(\ref{1.11}); the same holds in the LOE and LSE cases with the caveat that the definitions of the soft edge scaling and $N_{\rm s}'$ need to be modified (see \S\ref{S4.1}--\ref{S4.3} below).
The second, coming from \cite{Bo25}, is that the successive terms in the asymptotic expansions of the soft edge scaled densities of the Gaussian and Laguerre ensembles can be expressed in terms of 
transcendental basis functions --- of count three for the GUE and LUE, and five for the GOE, GSE, LOE, LSE --- with coefficients that are polynomials. In the case of the GUE, these
are
 \begin{equation}\label{1.17b}
 a_1(y) = ({\rm Ai}(y))^2, \quad  a_2(y) = ({\rm Ai}'(y))^2,
 \quad a_3(y) =  {\rm Ai}(y) {\rm Ai}'(y);
 \end{equation}
 thus for example from (\ref{1.15}), we have
 $ \rho_{(1)}^{\rm s}(y)  = - y  a_1(y) + a_2(y) $.

 \subsection{Scalar differential equation characterisation of the density and summary of our findings}
 The tool for our study of expansions of the form (\ref{1.14}) are linear third (GUE and LUE) and fifth (GOE, GSE, LOE, LSE) order differential equations satisfied by the density. A unified approach covering each of the Gaussian, Laguerre and Jacobi weights, valid for all even $\beta$ and their duals $4/\beta$ (the differential equation is of degree $\beta + 1$), has been given in \cite{RF21}. In Appendix~A, for $\beta =2$, we give a different unified approach applying to all the classical weights (the three listed above as well as the Cauchy weight) using orthogonal polynomial theory. For example, for the global scaled GUE density $\rho_{(1),N}^{\rm GUE,g}(x) := \sqrt{2N} \rho_{(1),N}^{\rm GUE}(\sqrt{2N} x)$ (the effect of global scaling is to map the leading order support to the interval $(-1,1)$), we have the differential equation
  \begin{equation}\label{1.3w}
  \bigg ( \Big ( {1 \over 4 N} \Big )^2 {d^3 \over d x^3} -
  (x^2 - 1) {d \over d x} + x \Big ) \bigg ) \rho_{(1),N}^{\rm GUE,g}(x) = 0.
  \end{equation}
  With a different interpretation (density of spinless free Fermions on a line, confined by a harmonic trap --- for a recent review on the relation of such quantum many body systems to random matrices, see \cite{DLMS19}), an equivalent third order differential equation characterisation first appeared in the 1979 work \cite{LM79}. Suppose now we know the existence of an expansion of $\rho_{(1),N}^{\rm GUE,g}(x)$ in inverse (possibly fractional) powers of $N$.
  Then, it follows from the differential equation that the expansion in fact only involves even powers of $N$. Moreover, the individual terms in the expansion are related by a sequence of nested inhomogeneous linear differential equations which all share the same homogeneous part.
  In the context of the Stieltjes transform (\ref{pkm}) below, this observation (and its consequences), can be found in \cite{HT12}.

  The above assumption of there being an expansion of $\rho_{(1),N}^{\rm GUE,g}(x)$ in inverse powers of $N$ is only for illustrative purposes: For the global density, all but the leading term contain oscillatory functions of $N$; see, e.g.,~\cite{GFF05}. Introducing the Stieltjes transform
  \begin{equation}\label{pkm}
  W_N^{\rm GUE}(x):= \int_{-\infty}^\infty
  \rho_{(1),N}^{\rm GUE,g}(y) (x - y)^{-1} \, dy,\quad x \notin \mathbb R,
  \end{equation}
  it is known from \cite{GT05,HT12,WF14} that this smoothed quantity satisfies the hypothesis and also the differential equation (\ref{1.3w}) but with the zero on the right-hand side replaced by $2N$.

It was recently noted in \cite{FRS25} that the soft edge scaling of (\ref{1.3w}) naturally manifests the expansion (\ref{1.14}). In Section \ref{S2}, we study said expansion using the general approach outlined in the previous paragraph. We find that exactly the particular solutions of the aforementioned nested inhomogeneous linear differential equations (with no contribution from the homogeneous solution) specify both $\rho_{(1)}^{\rm 1,s,G}(y)$ and $\rho_{(1)}^{\rm 2,s,G}(y)$, but not $\rho_{(1)}^{\rm 3,s,G}(y)$. This same property is found at the soft edge for the GOE, GSE (see \S\ref{S3.2}) and the Laguerre ensembles LUE (see \S\ref{S4.1}--\ref{S4.3}), LOE, LSE (see \S\ref{S5.1}). Likewise at the hard edge for the LUE (see \S\ref{S4.4}), but not the LOE (see \S\ref{S5.2}), where we instead find that the first correction when so characterised does involve a multiple of the limiting density. Here, the hard edge refers to a rescaling of the form $x\mapsto x/(4N_{\rm h}')$ with $N_{\rm h}'$ specified in (\ref{4.15}). An important point is that the solutions of both the homogeneous and inhomogeneous equations in all cases considered can be sought among a set of transcendental basis functions, as identified in \cite{Bo25} at the soft edge (recall the final paragraph of the previous subsection), and in the present work at the Laguerre hard edge. It was noted in \cite{FRS25} that the homogeneous part of the inhomogeneous linear differential equation for the bilateral Laplace transform of the soft edge scaled GUE density is first order. This allows for a simpler analysis, with each correction term having the structure (\ref{2.6}) below. Moreover, in equation (\ref{B.8}) of Appendix B, we give an integral formula based on a saddle point analysis which can (when aided by computer algebra, and subject to associated overhead constraints) produce the terms in the asymptotic expansion to any order.

 Our methods of analysing linear differential equations satisfied by the density are applied to the GOE and GSE in Section \ref{S3}. As already mentioned, this equation is now of degree five. 
 Nonetheless, the shifted variable $N_{\rm s}'$ 
 (recall the beginning of the paragraph containing (\ref{1.17b}))
 is clearly identifiable, and leads to  
 a finite (three terms in total) decomposition of this operator at the soft edge, with components proportional to $(N_{\rm s}')^{-2j/3}$ for $j=0,1,2$. Moreover,
 consistent with findings in \cite{Bo24}, upon the use of a simple scaling of the independent variable, each operator in this decomposition is the same for both the GOE and GSE.
 We can now proceed to check that the first and second corrections in the analogue of (\ref{1.14}) can be determined as particular solutions of the implied inhomogeneous linear equations, with no additive component from the solution of the homogeneous equation as commented on in the previous paragraph. The bilateral Laplace transform, for which the leading term and the first two corrections can be made explicit, also exhibits this property.
 In Remark \ref{R3.5} we comment on an analysis of the soft edge scaled density in the $\beta=6$ case of the Gaussian ensembles, from the viewpoint of the known seventh order linear differential equation for this quantity \cite{RF21}.

 The study of soft edge scaling for the Laguerre ensembles, undertaken in sections \ref{S4} and \ref{S5}, has from Bornemann's findings many similarities with their Gaussian counterparts (in particular sharing the same transcendental bases for the same symmetry class), although the expansion variable $N_{\rm s}'$ needs to be modified. In fact, the required modification depends on whether the soft edge results from the case of the Laguerre parameter $a$ fixed, or $a$ scaled as is natural in applications to multivariable statistics ($n,N$ as appearing in (\ref{1.2}) simultaneously going to infinity, with their ratio fixed). We note that the latter has two further subcases that must be distinguished --- the left and right soft edges.
 Also considered in these sections is the case of hard edge scaling, whereby eigenvalues in the neighbourhood of the origin are scaled to be of order unity apart. While an earlier study
 \cite{FT19} has identified the appropriate choice of $N_{\rm h}'$, 
 and another has noted asymptotic expansions in powers of
 $(N_{\rm h}')^{-2}$ \cite[Remark 4.3.1]{Fo25},
 missing from the early literature is the transendental bases replacing (\ref{1.17b}). We give such bases in terms of Bessel functions. Asymptotic analysis, based on the relation between Laguerre and hypergeometric polynomials, is performed to compute the second order correction term in the analogue of (\ref{1.14})
 for the hard edge scaled LUE (the first order term is known
 \cite{FT18}), and the first order correction term for the LOE.
 Surprisingly (keeping in mind experience of early calculations in our paper), for the latter it is found that from the viewpoint of the corresponding coupled linear differential equations, there is an additive multiple of the solution of the homogeneous equation.

\section{The GUE edge density} \label{S2}

\subsection{Previously established results}
Let $\rho_{(1),N}^{\rm GUE,s}(y)$ denote the left-hand side of (\ref{1.14}),
and thus the GUE density with soft edge scaling.
Starting with (\ref{1.3w})
we have shown in \cite[Eq.~(12)]{FRS25} that  this scaled density
satisfies the linear third order differential equation
\begin{equation}\label{2.1}
\Big ( {d^3 \over d y^3} - 4 y {d \over d y} + 2 \Big )
\rho_{(1),N}^{\rm GUE,s}(y) = {1 \over N^{2/3}} \Big ( y^2 {d \over d y} - y \Big ) 
\rho_{(1),N}^{\rm GUE,s}(y).
\end{equation}
In keeping with the recent finding in \cite{Bo25}
as revised in \S\ref{S1.2}, it follows immediately that $\rho_{(1),N}^{\rm GUE,s}(y)$ permits a large $N$  expansion in powers of $N^{-2/3}$,
\begin{equation}\label{2.2}
\rho_{(1),N}^{\rm GUE,s}(y) = r_0^{\rm G}(y) + 4^{-1} N^{-2/3} r_1^{\rm G}(y) +
4^{-2} N^{-4/3} r_2^{\rm G}(y) + \cdots;
\end{equation}
cf.~(\ref{1.14}).
A consequence (see \cite[Eq.~(13)]{FRS25}  with the scaling $r_j^{\rm G}(y) \mapsto 4^{-j} r_j^{\rm G}(y)$) is that the successive $r_j^{\rm G}(y)$ are related by the nested inhomogeneous, linear, third order differential equations
\begin{equation}\label{2.3}
\Big ( {d^3 \over d y^3} - 4 y {d \over d y} + 2 \Big ) r_j^{\rm G}(y)
 =  4 \Big ( y^2 {d \over d y} - y \Big ) 
r_{j-1}^{\rm G}(y), \quad j=0,1,\dots,
\end{equation}
with $r_{-1}^{\rm G}(y) = 0$ (note that the homogeneous parts of these differential equations are identical for all $j=0,1,\ldots$). The case $j=1$ was first recorded in
\cite[Eq.~(4.10)]{RF21}.

\begin{remark}
    From (\ref{1.17}), we know that $r_j^{\rm G}(y)$ for $j=1,2$ can be expressed as a $j$-dependent differential operator acting on $r_0^{\rm G}(y)$, a fact which moreover is expected to be true for general $j$. Using the first equation in (\ref{1.12}), in the case $j=1$, this relation reads
 \begin{equation}\label{2.3a}  
 {d \over dy} \bigg ( - {3 \over 10} {d \over dy} + {y^2 \over 5}  \bigg ) r_0^{\rm G}(y)= r_1^{\rm G}(y).
 \end{equation}
 This differential relation specifies the solution of 
  (\ref{2.3}) with $j=1$ (note that the solution of the latter is not unique without further information).
\end{remark}

Another finding of \cite{Bo25}
(which can be used to give a verification proof of (\ref{2.3a}))
is that the $ r_j^{\rm G}(y) $ have the structure
\begin{equation}\label{2.R}
 r_j^{\rm G}(y) = \alpha_j^{\rm G}(y)  ( {\rm Ai}(y)  )^2 +
  \beta_j^{\rm G}(y)  ( {\rm Ai}'(y)  )^2 + \gamma_j^{\rm G}(y)
   {\rm Ai}(y)  {\rm Ai}'(y)
\end{equation} 
for polynomials $\alpha^{\rm G}_j(y), \beta^{\rm G}_j(y), \gamma^{\rm G}_j(y)$ of degrees that depend on $j$. Note that since
\begin{align}\label{2.Ra}
\begin{aligned}
& {d \over dy} ( {\rm Ai}(y)  )^2 = 2  {\rm Ai}(y)  {\rm Ai}'(y), \\
& {d \over dy} ( {\rm Ai}'(y)  )^2 = 2 y {\rm Ai}(y) {\rm Ai}'(y), \quad
{d \over dy}  {\rm Ai}(y)  {\rm Ai}'(y) = ( {\rm Ai}'(y)  )^2 +
y ({\rm Ai}(y))^2,
\end{aligned}
\end{align} 
as follows from the rules of differentiation and the Airy differential equation,
this structure is closed with respect to differentiation.  We list in Table \ref{T1} low order cases as known from \cite{Fo93} ($j=0$),
\cite{GFF05} ($j=1$) and \cite{Bo25} ($j=2$).

\begin{table}[h!]
\centering
\begin{tabular}{c|c|c|c} 
$j$ & $\alpha_j^{\rm G}(y)$ & $\beta_j^{\rm G}(y)$ & $\gamma_j^{\rm G}(y)$ \\
\hline
  0  &  $-y$ & 1 & 0 \\[.3cm]
  1   & $- {3 \over 5} y^2$ & ${2 \over 5}y$ & ${3 \over 5}$ \\[.3cm]
  2   &  $  {39 \over 175} y^3 + {9 \over 100} $ & $-{3 \over 175}y^2$ & $-  {1 \over 25} y^4 - {99 \over 175} y $
  \end{tabular}
  \caption{Explicit form of low order cases of the coefficients in
  (\ref{2.R}).}
  \label{T1}
  \end{table}

In the usual notation associated with the Airy differential equation, it was noted in \cite{FRS25} that the  third order homogeneous differential equation obtained by setting $j=0$ in (\ref{2.3}) has the three linearly independent solutions
 \begin{equation}\label{2.4}
 ( {\rm Ai}'(y)  )^2 - y  ( {\rm Ai}(y)  )^2, \quad
  ( {\rm Bi}'(y)  )^2 - y  ( {\rm Bi}(y)  )^2, \quad
    {\rm Ai}'(y)  {\rm Bi}'(y)  - y   {\rm Ai}(y)  {\rm Bi}(y). 
 \end{equation} 
Due to their asymptotic behaviours (the first converges to zero as $y\to\infty$, while the latter two diverge to $-\infty$ in this limit), only the first of these is relevant to
(\ref{2.1}), which from Table \ref{T1} and (\ref{2.R}) is precisely $r_0^{\rm G}(y)$ (recall too (\ref{1.15})).
One can check that the equations (\ref{2.3}) for $j=0,1,2$ are satisfied by the explicit functional forms given by (\ref{2.R}) and Table \ref{T1}. Moreover, for $j=1,2$, $r_j^{\rm G}(y)$ is the particular solution which is linearly independent of the homogeneous solution. However, for $j=3$, it turns out that the particular solution contains an additive factor that is proportional to the first of the functional forms in (\ref{2.4}) and so is not uniquely determined.

 It was also noted in \cite{FRS25} that introducing the bilateral Laplace transform (this quantity played a key role in the analysis of \cite{Ok02} relating the GUE to particular intersection numbers),
 \begin{equation}\label{2.U}
 u_j(\gamma) := \int_{-\infty}^\infty e^{\gamma y} r_j^{\rm G}(y) \, dy, \quad {\rm Re} (\gamma) > 0,
  \end{equation} 
  leads to a simplification
 of (\ref{2.3}).  Thus, by multiplying both sides of (\ref{2.3}) by $e^{\gamma y}$ $({\rm Re}(\gamma) > 0)$, integrating over the range $y \in \mathbb R$, and simplifying using integration by parts assuming no contribution from the endpoints, it follows that
  \begin{equation}\label{2.5}
  4 \gamma u_j'(\gamma) + (6 - \gamma^3) u_j(\gamma) = -4 ( 
  \gamma u_{j-1}''(\gamma) + 3 u_{j-1}'(\gamma)), \quad u_{-1}(\gamma) :=0.
\end{equation}
In the case $j=0$ we see that the third order homogeneous equation has reduced to a first order homogeneous equation. This is consistent with the transform (\ref{2.U}) only being well defined for the first of the three linearly independent solutions (\ref{2.4}).

\subsection{Results relating to \texorpdfstring{$\{ u_j(\gamma) \}$}{\{uj(γ)\}}}
We consider first a direct evaluation, using (\ref{2.U}) and our knowledge of $r_j^{\rm G}(y)$ $(j=0,1,2)$. As a preliminary, we evaluate the bilateral Laplace transform of the Airy function products appearing in (\ref{2.R}).

\begin{proposition}\label{P1}
Let ${\rm Re} (\gamma) > 0$. We have
 \begin{align}
 &h_1(\gamma) :=  \int_{-\infty}^\infty e^{\gamma y}  ({\rm Ai}(y) )^2 \, dy =
 { e^{\gamma^3/12} \over 2 \sqrt{\pi \gamma}}, \: \: 
  h_2(\gamma) :=  \int_{-\infty}^\infty e^{\gamma y}  ({\rm Ai}'(y) )^2 \, dy = \Big ( {1 \over \gamma} + {d \over d \gamma} \Big ) h_1(y),
\nonumber \\ 
  & h_3(\gamma) :=  \int_{-\infty}^\infty e^{\gamma y}    {\rm Ai}(y) {\rm Ai}'(y)  \, dy = - {\gamma \over 2}  h_1(\gamma).
  \end{align} 
\end{proposition}

\begin{proof}
The first of these is returned by Mathematica, but this package fails on the remaining two. To obtain the second evaluation, we make use of integration by parts and the Airy differential equation. The third only requires integration by parts.
    \end{proof}

    \begin{corollary}\label{Co1}
    With ${\rm Re} (\gamma) > 0$, let $u_j(\gamma)$ be specified by (\ref{2.U}). We have
 \begin{align}\label{C1.1}
 & u_0(\gamma) = {e^{\gamma^3/12} \over 2 \sqrt{\pi} \gamma^{3/2}}, \quad   u_1(\gamma) = - {e^{\gamma^3/12} \over 160 \sqrt{\pi} \gamma^{5/2}} 
 \Big ( 60 + 20 \gamma^3 + \gamma^6 \Big ),
 \nonumber 
 \\ & u_2(\gamma) =  { e^{\gamma^3/12} \over 179200 \sqrt{\pi} \gamma^{7/2}} \Big ( -42000 + 28000 \gamma^3 + 14840 \gamma^6 + 680 \gamma^9 + 7 \gamma^{12}
 \Big ).
 \end{align}
 \end{corollary}

\begin{proof}
Substituting (\ref{2.R}) in  (\ref{2.U}) and integrating by parts shows that
 \begin{equation}\label{2.9a}  
u_j(\gamma) = \alpha_j^{\rm G} \Big ( {d \over d \gamma} \Big ) h_0(\gamma) +
 \beta_j^{\rm G} \Big ( {d \over d \gamma} \Big ) h_1(\gamma) +
 \gamma_j^{\rm G} \Big ( {d \over d \gamma} \Big ) h_2(\gamma). 
\end{equation}
For $j=0,1,2$, the polynomials $\alpha_j^{\rm G},\beta_j^{\rm G},\gamma_j^{\rm G}$ are known from Table \ref{T1}, while $h_0,h_1,h_2$ are known from Proposition
\ref{P1}. The result now follows from explicit computation.
    \end{proof}

    \begin{remark} ${}$ \\
    1.~In (\ref{2.U}), let us write $\gamma = \gamma^{\rm r} + i \gamma^{\rm i}$ with $\gamma^{\rm r}, \gamma^{\rm i} \in \mathbb R$ and $\gamma^{\rm r} > 0$.
Then, we see that $u_j(\gamma^{\rm r}+ i \gamma^{\rm i})$ is the Fourier transform of $e^{\gamma^{\rm r} y} r_j^{\rm G}(y)$, with Fourier variable
$ \gamma^{\rm i}$. Hence, by taking the inverse Fourier transform, it follows
 \begin{equation}\label{2.9b}  
e^{\gamma^{\rm r} y} r_j^{\rm G}(y) = {1 \over 2 \pi }
\int_{-\infty}^\infty e^{-iy x} u(\gamma^{\rm r} + i x) \, dx.
\end{equation}
One use of this formula is for the purpose of providing a plot of $r_j^{\rm G}(y)$ given $u_j(\gamma)$. For $j=0,1,2$, this can be compared against the same plots calculated from the exact functional forms known from (\ref{2.R}) and Table \ref{T1} --- agreement is obtained. \\
2.~Taking the bilateral Laplace transform of both sides of (\ref{2.3a}) shows that
 \begin{equation}\label{2.3au}  
-\bigg (  {3  \over 10} \gamma^2 + {1 \over 5}\gamma
{d^2 \over d \gamma^2} \bigg ) u_0(\gamma) =
u_1(\gamma),
\end{equation}
as can be verified from the explicit formulas in 
(\ref{C1.1}).
  \end{remark}

    We can check that the results of Corollary \ref{Co1} are consistent with the coupled differential equations (\ref{2.5}). Moreover, proceeding inductively, we see for general positive integer $j$ that these equations permit solutions
   \begin{equation}\label{2.6}  
   u_j(\gamma) =  {e^{\gamma^3/12} \over \sqrt{\pi} \gamma^{(2j+3)/2}}
   \sum_{l=0}^{2j} c_{l,j} \gamma^{3l},
   \end{equation}
   where the $c_{l,j}$ are rationals. For $j$ not a multiple of three, these solutions are unique as then  the solution to the homogeneous part of (\ref{2.5})
   (which is proportional to $u_0(\gamma)$) is not consistent with the functional form (\ref{2.6}), i.e., there is no $0\le l\le 2j$ such that the factor $\gamma^{3l-(2j+3)/2}$ in the $l$-th term of (\ref{2.6}) matches the equivalent factor of $\gamma^{-3/2}$ in $u_0(\gamma)$. In particular, this allows us to independently derive $u_1(\gamma), u_2(\gamma)$ in (\ref{C1.1}). In the case $j=3$,
   making use of knowledge of $u_2(\gamma)$ from (\ref{C1.1}) and substitution of (\ref{2.6}) gives that
  \begin{multline}\label{2.7}  
 u_3(\gamma) = -  {4^3 e^{\gamma^3/12} \over  \sqrt{\pi} \gamma^{9/2}}
 \bigg ( {105 \over 16384}     + b \gamma^3 + {1099 \over 196 600} \gamma^6 + {223 \over 229 376} \gamma^9 + {17089 \over 412 876 800} \gamma^{12}  \\ +{27 \over 45 875 200} \gamma^{15} + {1 \over 393 216 000} \gamma^{18} \bigg ),
 \end{multline}
 for some undetermined $b$. Note that the latter relates to the solution of the homogeneous part of (\ref{2.5}). 

 Using a direct method based on a saddle point analysis specified in Appendix B, we compute that
  \begin{equation}\label{2.8} 
  b = - {35 \over  16384}.
  \end{equation} 
  With $u_3(\gamma)$ now specified, taking the functional form (\ref{2.6}) as ansatz in
 (\ref{2.5}) then gives
 \begin{multline}\label{2.9}  
 u_4(\gamma) =  {4^4 e^{\gamma^3/12} \over  \sqrt{\pi} \gamma^{11/2}}
 \bigg ( - {4725\over 1 048 576}     + {315 \over 262 144} \gamma^3 + {3759 \over 1 048 576} \gamma^6 + {43 471 \over 22 020 096} \gamma^9 + {61 483 \over 293 601 280} \gamma^{12}  \\ +{113 941 \over 14 533 263 360} \gamma^{15} + {114 691 \over 924 844 032 000} \gamma^{18} 
 + {37 \over 44 040 192 000}  \gamma^{21} +
  {1 \over 503 316  480 000}  \gamma^{24} 
 \bigg ).
 \end{multline}
 Next, with knowledge of $u_4(\gamma)$, this process can be repeated to specify $u_5(\gamma)$. As already noted below (\ref{2.6}), and explicitly demonstrated in
 (\ref{2.7}), continuing this approach to the calculation of $u_6(\gamma)$
 there will be an additive factor proportional to $u_0(\gamma)$
 which is undetermined. In this case, the direct method of Appendix B could be applied to deduce the analogue of (\ref{2.8}).

\section{The GOE and GSE edge density} \label{S3}
\subsection{Previously established results}
In order to present the results in a compact form, here we adopt the convention used in \cite{Bo25} for the weight function: Instead of using $w^{\rm G}(x)= e^{-\beta x^2/2}$, we use $e^{-x^2/2}$ for the GOE case and $e^{-x^2}$ for the GSE case.  It has been 
demonstrated in \cite{Bo25}
that for $\beta=1,4$, by introducing a shifted variable
\begin{equation}\label{Nd}
N'_{\rm s}:=N+(\beta - 2)/(2 \beta), 
\end{equation}
under the soft edge scaling
\begin{align}\label{3.1}
    \rho_{(1),N}^{\mathrm{G}\beta\mathrm{E},s}(y)
:= {1 \over \sqrt{2} (N'_{\rm s})^{1/6}}
\rho_{(1),N} (\sqrt{2N'_{\rm s}} + y/(\sqrt{2} (N'_{\rm s})^{1/6}); \beta ),
\end{align}
the density of the GOE/GSE admits an asymptotic expansion in powers of $(N_{\rm s}')^{-2/3}$. Here, we introduce a further subscript label ``s'' denoting soft to distinguish from other definitions of $N'$ which occur at, for example, the hard edge to be considered below.  To present the results in a unified way, we write the expansion as
\begin{equation}\label{2.8a} 
 \beta^{\frac{1}{6}}\rho_{(1),N}^{\mathrm{G}\beta\mathrm{E},s}(y) =
 {}{r}_0^{\rm G,\beta}(\beta ^{\frac{1}{3}}y) + (8\sqrt{\beta}N'_{\rm s})^{-2/3} {}{r}_1^{\rm G,\beta}(\beta ^{\frac{1}{3}}y) +
(8\sqrt{\beta}N'_{\rm s})^{-4/3} {}{r}_2^{\rm G,\beta}(\beta ^{\frac{1}{3}}y) + \cdots.
 \end{equation}

\begin{remark} ${}$ \\
1.~In the notation of \cite{Bo25}, the expansion (\ref{2.8a})
is given in the variable $h_{n',\infty}=\frac{1}{4}\beta^{-\frac{1}{3}}(N_s')^{-\frac{2}{3}}$. \\
2.~The soft edge scaling of (\ref{3.1}) is said to be optimal, which means that no other choice of linear scaling $x=\sigma y+\mu$ will result in reducing the order of the correction (there is in fact an equivalence class of such optimal scalings). All edge scalings considered in this paper are optimal in this sense. Initial forays into the topic of optimal scaling \cite{Bo16,FT18} observed that the order of the correction could be minimised by re-centring the variable $y$ or, equivalently, fine-tuning the location of the edge by adding a microscopic shift. In the Gaussian case, for example, this meant changing the traditional scaling
\begin{equation} \label{scaling_old}
x = y/(\sqrt{2}N^{1/6}) + \sqrt{2N}
\end{equation}
into
\begin{equation} \label{scaling_shift}
x = y/(\sqrt{2}N^{1/6}) + \sqrt{2N} + (\beta-2)/(2\beta\sqrt{2N}).
\end{equation}
A Taylor expansion shows that under the scaling (\ref{scaling_old}), $\rho_{(1),N}^{\mathrm{G}\beta\mathrm{E},s}(y)$ has leading correction of order $N^{-1/3}$, with a functional form in $y$ proportional to the derivative of $r_0^{\rm G,\beta}(\beta ^{\frac{1}{3}}y)$; see \cite{FT18}, \cite[Remark 1.7]{BL24}, \cite[\S 1.3]{Bo25}. Taking a large $N$ expansion of the argument of (\ref{3.1}) shows that it is equivalent to (\ref{scaling_shift}) up to ${\rm O}(N^{-7/6})$. Thus, it can be seen as a further refinement of (\ref{scaling_shift}); compared to (\ref{scaling_old}), replacing $N$ by $N'_{\rm s}$ corresponds to changing both the centring and scaling parameters. The benefit of scaling according to (\ref{3.1}) is that one produces an expansion of the form (\ref{2.8a}), whereas no such expansion in integer powers of $N^{-2/3}$ is seen when applying the scaling (\ref{scaling_shift}). On a similar note, it is easy to see upon expanding the right-hand side of (\ref{2.8a}) in large $N$ that it has leading and next to leading order correction terms of order $N^{-2/3}$ and $N^{-4/3}$, but then there is a term of order $N^{-5/3}$, breaking the pattern. Finally, we emphasise that the phenomena just recounted have analogues for all cases considered in this paper.
\end{remark} 

Analogous to the GUE case, the functions ${r}_j^{\rm G,\beta}(y)$ are found to have the form of a linear combination involving the Airy function \cite{Bo25},
 \begin{align}\label{rGbeta}
      \begin{aligned}
          {}{r}_j^{\rm G,\beta}(y) =& {}{\alpha}_j^{\rm G,\beta}(y)  ( {\rm Ai}(y)  )^2 +
  {}{\beta}_j^{\rm G,\beta}(y)  ( {\rm Ai}'(y)  )^2 + {}{\gamma}_j^{\rm G,\beta}(y)
   {\rm Ai}(y)  {\rm Ai}'(y)\\
   &+{}{\xi}_j^{\rm G,\beta}(y){\rm Ai}(y){\rm AI}_\nu(y)+{}{\eta}_j^{\rm G,\beta}(y){\rm Ai}'(y){\rm AI}_\nu(y).
      \end{aligned}
 \end{align}
 Here, the polynomial coefficients ${}{\alpha}_j^{\rm G,\beta}(y),{}{\beta}_j^{\rm G,\beta}(y),{}{\gamma}_j^{\rm G,\beta}(y),{}{\xi}_j^{\rm G,\beta}(y),{}{\eta}_j^{\rm G}(y)$ are independent of $\beta$ and ${\rm AI}_\nu$ is the anti-derivative of the Airy function defined as
 \begin{align}\label{3.6}
     {\rm AI}_\nu(y):= \nu-\int_{y}^{\infty}{\rm Ai}(t)dt,
 \end{align}
 with
 \begin{align}\label{vb}
     \nu=\nu_\beta :=\left\{\begin{array}{cc}
         1, & \beta=1, \\
         0, & \beta=4.
     \end{array}
     \right.
 \end{align}
 Note that we have $\lim_{y\to-\infty}{\rm AI}_1(y)=\lim_{y\to\infty}{\rm AI}_0(y)=0 $. For the low order cases, the coefficients are given in Table \ref{T2}.
 \begin{table}[h!]
\centering
\begin{tabular}{c|c|c|c|c|c} 
$j$ & ${}{\alpha}_j^{\rm G,\beta}(y)$ & ${}{\beta}_j^{\rm G,\beta}(y)$ & ${}{\gamma}_j^{\rm G,\beta}(y)$ & ${}{\xi}_j^{\rm G,\beta}(y)$ & ${}{\eta}_j^{\rm G,\beta}(y)$\\
\hline
  0  &  $-y$ & 1 & 0 & $\frac{1}{2}$ & 0 \\[.3cm]
  1   & $- {1 \over 2} y^2$ & ${2 \over 5}y$ & ${3 \over 10}$ & $-{y\over 10}$ & ${y^2\over 10}$ \\[.3cm]
  2   &  $\frac{3 y^3}{25}+\frac{279}{700}$ & $-\frac{27 y^2}{350}$ & $-\frac{y^4}{100}-\frac{27 y}{140}$ & $\frac{y^5}{100}+\frac{9 y^2}{140}$ & $-\frac{3 y^3}{70}-\frac{9}{70}$
  \end{tabular}
  \caption{Explicit form of low order cases of the coefficients in
  (\ref{rGbeta}).}
  \label{T2}
  \end{table}

  \begin{remark}
The analogues of (\ref{1.12}) for the GOE and GSE are known from \cite[Eqns.~(5.2a) and (5.2b)]{Bo24}. Moreover, these operators permit a factorisation of the same form of that noted below (\ref{1.16}). From (\ref{1.17}) we then read off that
 \begin{equation}\label{2.3ar}  
  \bigg ( - {3 \over 5} {d^2 \over dy^2} + {y^2 \over 5} {d \over dy} +
 {2 \over 5} y \bigg ) {}{r}_0^{\rm G, \beta}(y)= {}{r}_1^{\rm G,\beta}(y)
 \end{equation}
 for both $\beta = 1$ and 4. 
 This result can be verified using Table \ref{T2} and the formulas (\ref{2.Ra}) supplemented by
  \begin{equation}\label{2.3aS} 
  {d \over d y} \Big ( {\rm Ai}(y) {\rm AI}_\nu(y) \Big ) = 
  {\rm Ai}'(y) {\rm AI}_\nu(y) +  ({\rm Ai}(y))^2, \: \:
   {d \over d y} \Big ( {\rm Ai}'(y) {\rm AI}_\nu(y) \Big ) = 
 y {\rm Ai}(y) {\rm AI}_\nu(y) +  {\rm Ai}'(y) {\rm Ai}(y).
   \end{equation}

   Consider now the general $\beta$ Gaussian ensemble 
   (weight $e^{-\beta x^2/2}$) and write the soft edge expansion of the density as
   \begin{equation}\label{2.8aG} 
 \rho_{(1),N}^{\mathrm{G}\beta\mathrm{E},s}(y) =
 {R}_0^{\rm G,\beta}(y) + (8N'_{\rm s})^{-2/3} {R}_1^{\rm G,\beta}(y) +
  \cdots.
 \end{equation}
 In comparison to~(\ref{2.2}) and (\ref{2.8a}), there is agreement for $\beta=1,2$ and thus  ${R}_j^{\rm G,\beta}(y) =
 {r}_j^{\rm G}(y)$ for $\beta = 2$, 
 and ${R}_j^{\rm G,\beta}(y) =
 {r}_j^{\rm G,\beta}(y)$ for $\beta = 1$.
 However, for $\beta =4$, since  the weight has been taken as $e^{-x^2}$ in (\ref{2.8a}) and as $e^{-2 x^2}$ in (\ref{2.8aG}), the expansions differ by a scaling.  Comparison of the results (\ref{2.3a}) and (\ref{2.3ar}) suggest that for general $\beta$,
  \begin{equation}\label{2.3aB}  
   \bigg ( - {3 \over 5 \beta } {d^2 \over dy^2} + {y^2 \over 5} {d \over dy} +
 {2 \over 5} y \bigg ) {R}_0^{\rm G, \beta}(y)= {R}_1^{\rm G,\beta}(y).
 \end{equation}
 For $\beta$ even, the density of the G$\beta$E has a $\beta$-dimensional integral form \cite{BF97a}. This was used in \cite{DF06} to establish an evaluation of $R_0^{\rm G, \beta}(y)$ as a $\beta$-dimensional integral, and in \cite{FT19a} to similarly express
 $R_1^{\rm G,\beta}(y)$ in terms of a $\beta$-dimensional integral. We expect that an integration by parts strategy as is common in studies of Selberg-type integrals (see \cite[\S 4.6]{Fo10}),  can be used to establish that these forms are related by (\ref{2.3aB}) (such a strategy was used in the recent work \cite{FS25} to establish an analogous identity for the circular $\beta$ ensemble).
 \end{remark}

\subsection{Coupled differential equations and the bilateral Laplace transform} \label{S3.2}

Without soft edge scaling, fifth order linear differential equations satisfied by
$\rho_{(1),N}^{\mathrm{GOE}}(x)$, $\rho_{(1),N}^{\mathrm{GSE}}(x)$ have been derived in \cite{WF14}. Changing to soft edge variables as specified in (\ref{3.1}) gives a terminating large $N'_{\rm s}$ form of the corresponding differential operators in powers of $(N'_{\rm s})^{-2/3}$.

\begin{proposition}\label{P3.2}
Introduce the differential operators
\begin{align}\label{3.11}
\begin{aligned}
& \mathcal{D}_{0}^{\mathrm{G},\beta}
:=\frac{4}{\beta}\frac{d^5}{d y^5}
-20 y \frac{d^3}{d y^3}
+12 \frac{d^2}{d y^2}
+16 \beta y^2 \frac{d}{d y}
-8 \beta y,\\[0.3em]
& \mathcal{D}_{1}^{\mathrm{G},\beta}
:=-5 y^2 \frac{d^3}{d y^3}
+6 y \frac{d^2}{d y^2}
+\Big(8\beta y^3+14-4\beta-\frac{16}{\beta}\Big)\frac{d}{d y}
-6\beta y^2, \quad 
\mathcal{D}_{2}^{\mathrm{G},\beta} 
:=\beta y^3\Big(y\frac{d}{d y}-1\Big).
\end{aligned}
\end{align}
Then, for $\beta=1$ and $4$, the soft edge scaled density satisfies
\begin{equation}\label{3.2}
\bigg(
\mathcal{D}_{0}^{\mathrm{G},\beta}
+\frac{1}{(N'_{\rm s})^{2/3}}\mathcal{D}_{1}^{\mathrm{G},\beta}
+\frac{1}{(N'_{\rm s})^{4/3}}\mathcal{D}_{2}^{\mathrm{G},\beta}
\bigg)
\rho_{(1),N}^{\mathrm{G}\beta\mathrm{E},s}(y)=0 .
\end{equation}
\end{proposition}

\begin{corollary}\label{C3.3}
    The terms in the expansion (\ref{2.8a}) satisfy
    the nested inhomogeneous differential equations
\begin{align}\label{eq-rgbeta}
    \mathcal{D}_0^{\mathrm{G}, \beta} {}{r}_j^{\mathrm{G}, \beta}(\beta^{\frac{1}{3}}y)= - \Big ( 4 \beta^{\frac{1}{3}}\mathcal{D}_1^{\mathrm{G}, \beta} {}{r}_{j-1}^{\mathrm{G}, \beta}(\beta^{\frac{1}{3}} y)+16\beta^{\frac{2}{3}}\mathcal{D}_2^{\mathrm{G}, \beta} {}{r}_{j-2}^{\mathrm{G}, \beta}(\beta^{\frac{1}{3}} y) \Big ), \quad
     {}{r}_{-1}^{\mathrm{G}, \beta}(y)= {}{r}_{-2}^{\mathrm{G}, \beta}(y)=0,
\end{align}
or equivalently,
\begin{align}\label{3.14E}
    \tilde{\mathcal{D}}_0^{\mathrm{G,\beta}} {}{r}_j^{\mathrm{G,\beta}}(y)+\tilde{\mathcal{D}}_1^{\mathrm{G},\beta} r_{j-1}^{\mathrm{G},\beta}(y)+\tilde{\mathcal{D}}_2^{\mathrm{G}, \beta} {}{r}_{j-2}^{\mathrm{G,\beta}}(y)=0,
\end{align}
where
\begin{align*}
    \begin{aligned}
& {}{\mathcal{D}}_{0}^{\mathrm{G,\beta}}
:=\frac{d^5}{d y^5}
-5 y \frac{d^3}{d y^3}
+3 \frac{d^2}{d y^2}
+4 y^2 \frac{d}{d y}
-2 y,\\[0.3em]
& {}{\mathcal{D}}_{1}^{\mathrm{G,\beta}}
:=-5 y^2 \frac{d^3}{d y^3}
+6 y \frac{d^2}{d y^2}
+\Big(8 y^3-6\Big)\frac{d}{d y}
-6 y^2,\quad 
{}{\mathcal{D}}_{2}^{\mathrm{G,\beta}}
:=4y^3\Big(y\frac{d}{d y}-1\Big).
\end{aligned}
\end{align*}
(Here, we have used the fact that $4\beta+\frac{16}{\beta}=20,\,\beta=1,4$ to simplify the formula.)
\end{corollary}

We see that the scaling of the variables has removed the $\beta$ dependence in the operators, making the corrections $r_j^{\rm G,\beta}(y) $ for the GOE/GSE described by the same system of equations. We can check from the differentiation formulas (\ref{2.Ra}) and (\ref{2.3aS}) that making an ansatz consistent with the polynomials in Table \ref{T2} for $j=1,2$ but with unknown coefficients, the latter is uniquely determined by the appropriate case of (\ref{3.14E}). Thus, as for the GUE, $r_j^{\mathrm{G}, \beta}(y)$ for $j=1,2$ are particular solutions distinct from the case $j=0$ which solves the corresponding homogeneous equation.

Similar to the GUE case, we introduce the bilateral Laplace transform
\begin{align}
    u_j^{\beta}(\gamma) := \int_{-\infty}^\infty e^{\gamma y} {}{r}_j^{\rm G, \beta}(y) \, dy, \quad {\rm Re} (\gamma) > 0.
\end{align}
Then from \eqref{3.14E}, we compute that $u_j^{\beta}(\gamma)$ satisfies
\begin{equation} \label{3.15}
\mathcal L_0^{\rm G,\beta}u_j^{\beta}(\gamma)+\mathcal L_1^{\rm G,\beta}u_{j-1}^{\beta}(\gamma)+\mathcal L_2^{\rm G,\beta}u_{j-2}^{\beta}(\gamma)=0,\quad u_{-1}^{\beta}(\gamma)=u_{-2}^{\beta}(\gamma)=0,
\end{equation}
where
\begin{align}\label{3.16}
\begin{aligned}
&\mathcal L_0^{\rm G,\beta}
=
18\gamma^2-\gamma^5
+5(\gamma^3-2)\frac{d}{d\gamma}
-4\gamma\frac{d^2}{d\gamma^2}
,\\[0.3em]
&\mathcal L_1^{\rm G,\beta}
=48\,\gamma
+36\gamma^2\frac{d}{d\gamma}
-5(6-\gamma^3)\frac{d^2}{d\gamma^2}
-8\gamma\frac{d^3}{d\gamma^3},\quad
\mathcal L_2^{\rm G,\beta}
=
-20\frac{d^3}{d\gamma^3}
-4\gamma\frac{d^4}{d\gamma^4}
.
\end{aligned}
\end{align}
Especially, the leading term satisfies a second order homogeneous equation
\begin{align}\label{3.17}
\bigg(18\gamma^{2}-\gamma^{5}\bigg)\,u_0^{\beta}(\gamma)
+5\bigl(\gamma^{3}-2\bigr)\,\frac{d}{d\gamma}u_0^{\beta}(\gamma)
-4\gamma\,\frac{d^2}{d\gamma^2}u_0^{\beta}(\gamma)
=0.
\end{align}
For this, one can verify that the two fundamental solutions (returned by the computer algebra package Mathematica) are 
given by
\begin{align}\label{3.18}
\exp\!\left(\frac{\gamma^3}{3}\right),\quad 
\frac{\exp\!\left(\frac{\gamma^3}{12}\right)}{3\,\gamma^{3/2}}
\left(
2
+\exp\!\left(\frac{\gamma^3}{4}\right)\sqrt{\pi}\,\gamma^{3/2}\,
\operatorname{Erf}\!\left(\frac{\gamma^{3/2}}{2}\right)
\right).
\end{align}
Supplementary to Proposition \ref{P1}, we evaluate the bilateral Laplace transform of the relevant basis functions ${\rm Ai}{\rm AI}_\nu $ and ${\rm Ai}'{\rm AI}_\nu$.
\begin{proposition}
    Let ${\rm Re} (\gamma) > 0$. For $\nu=0,1$ we have
    \begin{align}
    &h_4(\gamma):=\int_{-\infty}^{\infty} e^{\gamma y} {\rm Ai}(y) {\rm A I}_{\nu}(y) d y=\frac{1}{2} e^{\frac{\gamma^3}{3}}\left({\rm Erf}\left(\frac{\gamma^{3 / 2}}{2}\right)-1\right)+\nu e^{\frac{\gamma^3}{3}},\\
    &h_5(\gamma):=\int_{-\infty}^{\infty} e^{\gamma y} {\rm Ai}'(y) {\rm A I}_{\nu}(y) d y=-h_1(\gamma)-\gamma h_4(\gamma).
\end{align}
\end{proposition}
\begin{proof}
    The first evaluation is obtained by integration by parts and solving a first order differential equation. The constant is fixed by the equation $\int_{-\infty}^{\infty}{\rm Ai}(y){\rm AI}_1(y)dy={ 1 \over 2} $. The second only requires integration by parts.
\end{proof}
As an immediate result, we can deduce explicit forms for $u_j^{\beta}(\gamma) $ with small $j$ from that of $r_j^{\rm G,\beta}(y) $:
\begin{align}
   \begin{aligned}
        &u_0^{\beta}(\gamma)=\frac{1}{4}(2 \nu_{\beta}-1) e^{\gamma^3/3}+\frac{e^{\gamma^3/12}}{4 \sqrt{\pi} \gamma^{3 / 2}}\left(2+\sqrt{\pi} e^{\frac{\gamma^3}{4}} \gamma^{\frac{3}{2}} \operatorname{Erf}\left(\frac{\gamma^{\frac{3}{2}}}{2}\right)\right),\\
&u_1^{\beta}(\gamma)=-\frac{e^{\gamma^{3}/12}}{80\sqrt{\pi}\,\gamma^{5/2}}
\left(30+25\gamma^{3}+8\gamma^{6}\right)
-\frac{e^{\gamma^{3}/3}}{20}\,
\gamma^{2}\left(5+\gamma^{3}\right)
\left(-1+2\nu_\beta +\operatorname{Erf}\!\left(\frac{\gamma^{3/2}}{2}\right)\right),\\
&u_2^{\beta}(\gamma)=\frac{e^{\gamma^{3}/12}}{89600\,\sqrt{\pi}\,\gamma^{7/2}}
\left(
-21000
+37100\,\gamma^{3}
+79870\,\gamma^{6}
+19975\,\gamma^{9}
+896\,\gamma^{12}
\right)
\\
&\qquad +
\frac{e^{\gamma^{3}/3}}{1400}\,
\gamma
\left(
700
+875\,\gamma^{3}
+170\,\gamma^{6}
+7\,\gamma^{9}
\right)
\left(
-1+2\nu_\beta+\operatorname{Erf}\!\left(\frac{\gamma^{3/2}}{2}\right)
\right).
   \end{aligned}
\end{align}
Here, one sees that $u_0^\beta(\gamma)$ is a linear combination of the two fundamental solutions (\ref{3.18}) of (\ref{3.17}), and that $u_1^\beta(\gamma), u_2^\beta(\gamma)$ are particular solutions of the corresponding cases of the inhomogeneous differential equation (\ref{3.15}),
which do not contain an additive term proportional to either of the fundamental solutions. This property of $u_1^\beta(\gamma), u_2^\beta(\gamma)$ is then the same as that observed for $\beta = 2$
(GUE case).

\begin{remark} \label{R3.5} From \cite[\S 2.3]{RF21}, we have available a seventh order linear differential equation for the density of the Gaussian $\beta$ ensemble with $\beta =6$ and its dual $\beta = 2/3$. Changing variables to a soft edge scaling, this differential equation can be shown to be of a form analogous to (\ref{2.1}) and (\ref{3.2}), and thus to permit 
the solution (\ref{2.8aG}), which expands in powers of $1/(N'_{\rm s})^{2/3}$. Upon using computer algebra to follow the methodology of \cite{RF21} further, this observation can also be made for $\beta=8,10,12$ (the relevant differential equations are of order $\beta+1$) and their duals $\beta=1/2,2/5,1/3$. As an explicit example, for $\beta = 6$ 
and $j=0,1$, we have
\begin{equation}\label{3.23}
\mathcal{D}_0^{\mathrm{G}, 6} {r}_j^{\mathrm{G}, 6}(y)= - \mathcal{D}_1^{\mathrm{G}, 6} {r}_{j-1}^{\mathrm{G}, 6}(y), \quad   {r}_{-1}^{\mathrm{G}, 6}(y) = 0,
\end{equation}
where $\mathcal{D}_0^{\mathrm{G}, 6}$ is as reported in 
\cite[Th.~4.1]{RF21}, while a computer algebra aided calculation gives 
\begin{equation*}
\mathcal{D}_1^{\mathrm{G}, 6}=-42y^2\frac{d^5}{dy^5}+42y\frac{d^4}{dy^4}+12(98y^3-9)\frac{d^3}{dy^3}-1404y^2\frac{d^2}{dy^2}-324y(16y^3-5)\frac{d}{dy}+3456y^3-234;
\end{equation*}
equivalent forms for $\beta=2/3$ result from appropriate scaling. However, for these $\beta$, there is no evidence that the limiting density can be expressed in terms of Airy functions as known for $\beta = 1,2$ and $4$. Instead, from \cite{DF06}, there is a six-dimensional integral form for $\beta =6$, and no known explicit functional form for $\beta = 2/3$. Thus, there is no scope to express the functional forms of corrections in terms of a basis of explicit transcendental functions.
\end{remark}

\section{The LUE edge densities} \label{S4}

We know from \cite{GT05} and \cite{RF21} that the LUE density
$\rho_{(1), N}^{\rm LUE}(x)$ (note that in this notation, the dependence on the Laguerre parameter $a$ has been suppressed) satisfies a third order homogeneous linear differential equation
$ \mathcal{D}^{\rm LUE} \rho_{(1),  N}^{\rm LUE}(x)=0$,
where
\begin{align}\label{4.0}
    \mathcal{D}^{\rm LUE}:=x^3 {d^3 \over d x^3}+4 x^2 {d^2 \over d x^2}-\left[x^2-2(a+2 N) x+a^2-2\right] x {d \over d x} +\left[(a+2 N) x-a^2\right].
\end{align}
The analysis of this at the soft edge depends on the parameter $a$ being fixed, or being proportional to $N$.

\subsection{Soft edge --- the case \texorpdfstring{$a$}{a} fixed} \label{S4.1}

Considering first the case that $a$ is fixed relative to $N$, let us now introduce the particular soft edge scaling variable $y$ by
\begin{equation}\label{4.1}
x = 4N + 2a   + 2 (2N+ a )^{1/3} y;
\end{equation}
cf.~\cite[\S 7.2.2]{Fo10} with the shift $N \mapsto N+a/2$ motivated by a finding of the study \cite{FT19}.

With $N'_{{\rm s},a}:= N + a/2$, after dividing through by $(2N'_{{\rm s},a})^2$, a (computer algebra aided) calculation gives that the differential operator (\ref{4.0}) then has the terminating large $N'_{{\rm s},a}$ expansion
\begin{multline}\label{4.2}
{1 \over (2N'_{{\rm s},a})^2} \mathcal{D}^{\rm LUE}
\Big |_{x \mapsto 4N'_{{\rm s},a} + 2 (2N'_{{\rm s},a})^{1/3} y}
\\ = \mathcal{D}_{0}^{{\rm LUE, s},a} +
{1 \over (2N'_{{\rm s},a})^{2/3}}\mathcal{D}_{1}^{{\rm LUE, s},a} +{1 \over (2N'_{{\rm s},a})^{4/3}}\mathcal{D}_{2}^{{\rm LUE,s},a} + {1 \over (2N'_{{\rm s},a})^{2}}\mathcal{D}_{3}^{{\rm LUE,s},a},
\end{multline}
where
\begin{align}\label{4.3}
\begin{aligned}
& \mathcal{D}_{0}^{{\rm LUE,s},a} := {d^3 \over d y^3} - 4 y {d \over dy} + 2,
\quad
  \mathcal{D}_{1}^{{\rm LUE,s},a} := 3 y {d^3 \over d y^3} +
4 {d^2 \over dy^2} - 8 y^2 {d \over dy} + 2 y, \\
&
 \mathcal{D}_{2}^{{\rm LUE,s},a} := 3 y^2 {d^3 \over d y^3} +
8 y {d^2 \over dy^2} - (4 y^3 +  (a^2 - 2)) {d \over dy}, \\
& \mathcal{D}_{3}^{{\rm LUE,s},a} :=   y^3 {d^3 \over d y^3} +
4 y^2 {d^2 \over d y^2} - (a^2 - 2)  y{d \over dy} - a^2.
\end{aligned}
\end{align}
Note that $\mathcal{D}_{0}^{{\rm LUE,s},a}$ is the same operator as appearing on the left-hand side of (\ref{2.1}).
Setting $\rho_{(1),n}^{{\rm LUE,s},a}(y) := (2N'_{{\rm s},a})^{1/3} \rho_{(1), N}^{\rm LUE}(4 N'_{{\rm s},a} + 2 (2 N'_{{\rm s},a})^{1/3} y)$, the decomposition (\ref{4.2}) shows that the differential equation admits solutions with a large $N'_{{\rm s},a}$ asymptotic expansion of the form
\begin{equation}\label{4.4}
\rho_{(1),N}^{{\rm LUE,s},a}(y) = r_0^{{\rm L,s},a}(y) + (2N'_{{\rm s},a})^{-2/3} r_1^{{\rm L,s},a}(y) +
(2 N'_{{\rm s},a})^{-4/3} r_2^{{\rm L,s},a}(y) + \cdots,
\end{equation}
since (\ref{4.2}) depends on $N_{{\rm s},a}'$ only through integer powers of $(N_{{\rm s},a}')^{-2/3}$; cf.~(\ref{2.2}). From the earlier work \cite{GFF05}, we know that $r_0^{{\rm L,s},a}(y) =
r_0^{\rm G}(y)$ as specified by (\ref{1.15}), in keeping with the universality of the soft edge limit for matrix ensembles with unitary symmetry. In particular, this quantity is independent of $a$. Analogous to (\ref{2.3}), we have that successive correction terms are linked by coupled inhomogeneous differential equations.

\begin{proposition}
    For $k=0,1,\dots$, we have
   \begin{equation}\label{4.5} 
 \mathcal{D}_{0}^{{\rm LUE,s},a} r_k^{{\rm L,s},a}(y) = - \Big (
 \mathcal{D}_{1}^{{\rm LUE,s},a}   r_{k-1}^{{\rm L,s},a}(y) +
  \mathcal{D}_{2}^{{\rm LUE,s},a} r_{k-2}^{{\rm L,s},a}(y)   +
   \mathcal{D}_{3}^{{\rm LUE,s},a} r_{k-3}^{{\rm L,s},a}(y) \Big ),
\end{equation}
with $r_{-3}^{{\rm L,s},a}(y) = r_{-2}^{{\rm L,s},a}(y) = r_{-1}^{{\rm L,s},a}(y) = 0$.
   \end{proposition}

   From \cite{Bo25}, we know that the $r_j^{{\rm L,s},a}(y)$ admit the same structural expansion (\ref{2.R}) as for the $r_j^{\rm G}(y)$. 
   With regards to $r_1^{{\rm L,s},a}(y)$, we observe from (\ref{4.3}) that the differential operator $ \mathcal{D}_{1}^{{\rm LUE,s},a} $ is independent of the parameter $a$. It then follows from (\ref{4.5}) with $k=1$ that
   $r_1^{{\rm L,s},a}(y)$ is similarly independent of $a$. It turns out that the case $a=0$ of $r_1^{{\rm L,s},a}(y)$ is known from
   \cite[Eq.~(3.10a) with $\tau = 1$]{Bo24}, \cite[Th.~4.3 with $\tau = 1$]{Bo25} (see also Table \ref{T3} below). Hence, for all $a>-1$,
   \begin{equation}\label{4.6}  
   r_1^{{\rm L,s},a}(y) = {1 \over 5} \Big ( 3 y^2 ({\rm Ai}(y))^2 - 2 y
   ({\rm Ai}'(y))^2 + 2 {\rm Ai}(y) {\rm Ai}'(y) \Big ).
\end{equation}   
We can check that this functional form is consistent with the inhomogeneous differential equation 
(\ref{4.5}) with $k=1$. Now that we have knowledge of $r_0^{{\rm L,s},a}(y),
r_1^{{\rm L,s},a}(y)$, we can furthermore use (\ref{4.5}) with $k=2$ to specify the inhomogeneous differential equation satisfied by $r_2^{{\rm L,s},a}(y)$. We seek a solution of this equation with the structure (\ref{2.R}), and for the polynomials therein having the same form as those in Table \ref{T1} with $j=2$,
 \begin{equation}\label{4.7} 
 \alpha_2^{\rm L}(y) = a_3 y^3 + a_0, \quad \beta_2^{\rm L}(y) = b_2 y^2, \quad \gamma_2^{\rm L}(y) = c_4 y^4 + c_1 y.
\end{equation} 
Such a solution is unique as it does not contain an additive multiple of $r_0^{\rm L,s}$, which is the solution of the homogeneous portion of the differential equation. Carrying out the calculation gives
 \begin{equation}\label{4.8} 
 r_2^{{\rm L,s},a}(y) = \Big (- {96 \over 175} y^3 +  {4 - 2 a^2 \over 100}
 \Big ) ({\rm Ai}(y))^2  + {37 \over 175} y^2 ({\rm Ai}'(y))^2 -
 \Big (  {1 \over 25} y^4 + {74 \over 175} y \Big )  {\rm Ai}(y) {\rm Ai}'(y).
\end{equation}  
As with $r_1^{{\rm L,s},a}(y)$, the case $a=0$ of $r_2^{{\rm L,s},a}(y)$ is available in  \cite[Th.~4.3 with $\tau = 1$]{Bo25} (see also Table \ref{T3} below), and agreement is found with this specialisation of (\ref{4.8}).

Important in relation to the analogue of the first equation in (\ref{1.12}) for the fixed-$a$ Laguerre soft edge generating function expansion
 \begin{equation}\label{4.8x} 
 \mathcal E_{N,2}((4 N'_{{\rm s},a} + 2 (2N'_{{\rm s},a})^{1/2} y,\infty); x^a e^{-x};\xi) =
  \mathcal E_2^{{\rm s}} ((y,\infty);\xi) + {1 \over (2 N'_{{\rm s},a})^{2/3}} 
    \mathcal E_2^{1,{\rm s},a} ((y,\infty);\xi) + \cdots 
 \end{equation}
 is a second order differential operator mapping $r_0^{{\rm L,s},a}(y)$ to
 $r_1^{{\rm L,s},a}(y)$. Thus, we seek the fixed-$a$ Laguerre analogue of 
 (\ref{2.3a}) for the GUE.

 \begin{proposition}
     We have
      \begin{equation}\label{4.8y}
      {d^2 \over d y^2} r_0^{{\rm L,s},a}(y) = - 2 {\rm Ai}(y) {\rm Ai}'(y), \quad   {d \over d y} r_0^{{\rm L,s},a}(y) = - ( {\rm Ai}(y) )^2,
 \end{equation} 
 and consequently,
 \begin{equation}\label{4.8z}  - {1 \over 5} {d \over d y} \Big ( {d \over d y} 
 + y^2 \Big ) r_0^{{\rm L,s},a}(y) = r_1^{{\rm L,s},a}(y).
  \end{equation} 
 \end{proposition}

 \begin{proof}
     To derive (\ref{4.8y}), we make use of (\ref{1.15}) and (\ref{2.Ra}).
     With this established, (\ref{4.8z}) can be verified after recalling
     that $r_1^{{\rm L,s},a}(y)$ is given by (\ref{4.6}).
 \end{proof}
   
From the assumption that $ \mathcal E_2^{\rm s} ((y,\infty);\xi)$ and
$ \mathcal E_2^{1,{\rm s},a} ((y,\infty);\xi)$ are related by a $\xi$-independent second order differential operator, it follows from (\ref{4.8z}) that
 \begin{equation}\label{4.8z1}  - {1 \over 5}  \Big ( {d \over d y} 
 + y^2 \Big ) {d \over d y} 
  \mathcal E_2^{\rm s} ((y,\infty);\xi) =  \mathcal E_2^{1,{\rm s},a} ((y,\infty);\xi).
  \end{equation} 
  This can be compared against the result of \cite[Eq.~(3.7a) with $\tau = 1$]{Bo24}, which corresponds to the case $a=0$ (note that
  (\ref{4.8z1}) is independent of $a$). Agreement is found.

\subsection{Right soft edge --- the case \texorpdfstring{$a$}{a} proportional to \texorpdfstring{$N$}{N}}\label{S4.2}

For $\beta = 2$, we have from (\ref{1.2}) that $a = n - N$.
Here, we set $n = \gamma N$, $\gamma \ge 1$ so that $a = (\gamma - 1) N$, and is thus 
proportional to $N$, as is consistent with earlier work \cite{Jo01,Bo24}. We moreover introduce the soft edge scaling variable $y$ according to
\begin{align}\label{vJ}
    x=(1+\sqrt{\gamma})^{2}\,N_{\rm s}'+\gamma^{-1/6}(1+\sqrt{\gamma})^{4/3}(N_{\rm s}')^{1/3} y.
\end{align}
(For $\beta = 2$, $N_{\rm s}' = N$; using $N_{\rm s}'$ in
(\ref{vJ}) is required later when the cases $\beta =1$ and
4 are considered.)
It is further convenient to introduce the modification of $N_{\rm s}'$,
\begin{align}\label{Nh}
    \hat{N}_{\rm s}'=(4\gamma)^{1/2}\tau N_{\rm s}', \quad \tau=\frac{4}{\sqrt{\gamma}+1/\sqrt{\gamma}+2}.
\end{align}
In the work \cite{Bo25}, the expansion parameter $h_{n,p}$, when specified at the right soft edge, relates to our $\hat{N}_{\rm s}'$ by $
    h_{n,p}=(\hat{N}_{\rm s}')^{-2/3}$ for $\beta=2$ and $
    h_{n,p}=(\sqrt{\beta}\hat{N}_{\rm s}')^{-2/3}$ for $\beta=1,4$. After a change of variables in \eqref{4.0}, the LUE density scaled at the right soft edge with $a$ proportional to $N$, denoted by $\rho^{{\rm LUE},sr}_{(1),N}(y)$, can be checked to satisfy the differential equation
\begin{equation}\label{4.17r}
\Bigg(
(\hat N_{\rm s}')^{2}\,{}{\mathcal D}^{\mathrm{LUE},sr}_{0}
+(\hat N_{\rm s}')^{4/3}\,{}{\mathcal D}^{\mathrm{LUE},sr}_{1}
+(\hat N_{\rm s}')^{2/3}\,{}{\mathcal D}^{\mathrm{LUE},sr}_{2}
+{}{\mathcal D}^{\mathrm{LUE},sr}_{3}
\Bigg)\rho^{\mathrm{LUE},sr}_{(1),N}(y)=0,
\end{equation}
where
\begin{align} \label{4.17X}
\begin{aligned}
{\mathcal D}^{\mathrm{LUE},sr}_{0}
={}&
\frac{d^{3}}{dy^{3}}
-4y\frac{d}{dy}
+2, \quad 
{\mathcal D}^{\mathrm{LUE},sr}_{1}
={}
\tau\left(3y\frac{d^{3}}{dy^{3}}
+4\frac{d^{2}}{dy^{2}}-4y^{2}\frac{d}{dy}-2y\right)
-4\left(y^{2}\frac{d}{dy}-y\right),  \\[0.6em]
{\mathcal D}^{\mathrm{LUE},sr}_{2}
={}&
\tau^2\bigg(3y^{2}\frac{d^{3}}{dy^{3}}
+8y\frac{d^{2}}{dy^{2}}
+2\frac{d}{dy}\bigg)-4\tau y^{3}\frac{d}{dy},
\quad
{\mathcal D}^{\mathrm{LUE},sr}_{3}
={}
\tau^3\bigg(y^{3}\frac{d^{3}}{dy^{3}}
+4y^{2}\frac{d^{2}}{dy^{2}}
+2y\frac{d}{dy}\bigg).
\end{aligned}
\end{align}
As expected by the universality of the limiting soft edge density, the differential operator ${\mathcal D}^{\mathrm{LUE},sr}_{0}$ is identical to that on the left-hand side of (\ref{2.1}).
Now introducing the asymptotic expansion in powers of $(\hat{N}'_{\rm s})^{-2/3}$
established in \cite{Bo25},
\begin{align}\label{4.18X}
    \rho^{{\rm LUE},sr}_{(1),N}(y) =r_0^{{\rm L,sr}}(y)+r_1^{{\rm L,sr}}(y)(\hat{N}_{\rm s}')^{-2/3}+r_2^{{\rm L,sr}}(y)(\hat{N}_{\rm s}')^{-4/3}+\dots,
\end{align}
coupled differential equations for $\{ r_k^{{\rm L,sr}}(y) \}$ are immediate.

\begin{proposition}\label{P4.3X}
We have 
\begin{align}\label{4.16a}
    \begin{aligned}
    &{\mathcal D}^{\mathrm{LUE},sr}_{0}r_k^{{\rm L,sr}}(y)+{\mathcal D}^{\mathrm{LUE},sr}_{1}r_{k-1}^{{\rm L,sr}}(y) +{\mathcal D}^{\mathrm{LUE},sr}_{2}r_{k-2}^{{\rm L,sr}}(y)
+{\mathcal D}^{\mathrm{LUE},sr}_{3}r_{k-3}^{{\rm L,sr}}(y)
=0,
\end{aligned}
\end{align}
where $r_{-3}^{{\rm L,sr}}(y)=r_{-2}^{{\rm L,sr}}(y)=r_{-1}^{{\rm L,sr}}(y)=0$. 
\end{proposition}

From (\ref{4.16a}), it can be seen that unlike the fixed parameter case, the first correction term $r_{1}^{{\rm L,sr}}(y)$ now depends on the Laguerre parameter $a$ (through dependence on $\tau$). 
As for the Gaussian case, it is shown in \cite{Bo25} that the $r_j^{{\rm L,sr}}(y) $ have the structure
\begin{align}\label{rjLUEsr}
    r_j^{{\rm L,sr}}(y)= \alpha_j^{{\rm L,sr}}(y)  ( {\rm Ai}(y)  )^2 +
  \beta_j^{{\rm L,sr}}(y)  ( {\rm Ai}'(y)  )^2 + \gamma_j^{{\rm L,sr}}(y)
   {\rm Ai}(y)  {\rm Ai}'(y).
\end{align}
The low order coefficients, as calculated in \cite{Bo25}, are listed in Table \ref{T3}. 

\begin{table}[h!]
\centering
\renewcommand{\arraystretch}{1.2}
\begin{tabular}{c|c|c|c}
$j$ 
& $\alpha_j^{\rm L,sr}(y)$ 
& $\beta_j^{\rm L,sr}(y)$ 
& $\gamma_j^{\rm L,sr}(y)$ \\
\hline
0 
& $-y$ 
& $1$ 
& $0$ \\[0.2cm]

1 
& $\frac{3(2{\tau}-1)}{5}\,y^2$
& $-\frac{2(2{\tau}-1)}{5}\,y$
& $\frac{3-{\tau}}{5}$
\\[0.4cm]

2 
& $-\frac{214{\tau}^2-79{\tau}-39}
{175}\,y^3
+\frac{({\tau}-3)^2}
{100}$
&
$\frac{143{\tau}^2-103{\tau}-3}
{175}\,y^2$
&
$-\frac{(2{\tau}-1)^2}
{25}\,y^4
+\frac{29{\tau}^2-4{\tau}-99}
{175}\,y$
\end{tabular}
\caption{Explicit form of low order cases of the coefficients in \eqref{rjLUEsr}. Setting $\tau = 0$ reclaims the GUE results of Table \ref{T1} (see \cite[Remark 4.1]{Bo25}).}
\label{T3}
\end{table}

\noindent
Our point in relation to these explict functional forms and Proposition \ref{P4.3X} is that, as for the case of $a$ fixed considered in the previous subsection, both
$r_1^{{\rm L,sr}}(y)$ and $r_2^{{\rm L,sr}}(y)$ are particular solutions of the respective inhomogeneous differential equations (\ref{4.16a}). 

\begin{remark}
Analogous to (\ref{4.8z}), as noted in \cite{Bo25}, the corrections
$r_1^{\rm L,sr}(y)$, $r_2^{\rm L,sr}(y)$ can be related to
$r_0^{\rm L,sr}(y)$ by particular differential operators (of degree two and four respectively). Explicitly, for $r_1^{\rm L,sr}(y)$, as 
 can be checked using Table \ref{T3} and (\ref{1.15}),
 we have 
 \begin{equation}\label{4.8zA}  
 {1 \over 5 } {d \over d y}  \bigg ( {\tau - 3 \over 2} {d \over d y} 
 - (2 \tau - 1)  y^2 \bigg ) r_0^{\rm L,sr}(y) = r_1^{\rm L,sr}(y).
  \end{equation} 
  In keeping with \cite[Remark 4.1]{Bo25},  setting $\tau = 0$ reclaims (\ref{2.3a}).
\end{remark}

\subsection{Left soft edge --- the case \texorpdfstring{$a$}{a} proportional to \texorpdfstring{$N$}{N}}\label{S4.3}
In the context of the LUE, the left soft edge was first considered in \cite{BFP98}. With the same notation as in the last subsection, it is
realized by the scaling variable
\begin{align*}
    x=(1-\sqrt{\gamma})^{2}\,N_{\rm s}'-\gamma^{-1/6}(\sqrt{\gamma}-1)^{4/3}(N_{\rm s}')^{1/3} y.
\end{align*}
To fit the left soft edge case, we also modify the parameters
\begin{align}
    \hat{N}_{\rm s,l}' :=(4\gamma)^{1/2}\tau_{\rm l} N_{\rm s}',\quad {\tau}_{\rm l}:=\frac{4}{\sqrt{\gamma}+1/\sqrt{\gamma}-2}.
\end{align}
Denote by $\rho^{\mathrm{LUE,sl}}_{(1),N}(y)$ the left soft edge scaled LUE density. Starting with {\eqref{4.0}, and proceeding as in the calculation of (\ref{4.17r}), we find that
\begin{equation}\label{4.24}
\Bigg(
(\hat N_{\rm s,l}')^{2}\,{}{\mathcal D}^{\mathrm{LUE,sl}}_{0}
+(\hat N_{\rm s,l}')^{4/3}\,{}{\mathcal D}^{\mathrm{LUE,sl}}_{1}
+(\hat N_{\rm s,l}')^{2/3}\,{}{\mathcal D}^{\mathrm{LUE,sl}}_{2}
+{}{\mathcal D}^{\mathrm{LUE,sl}}_{3}
\Bigg)\rho^{\mathrm{LUE,sl}}_{(1),N}(y)=0,
\end{equation}
where
\begin{align}\label{4.25}
\begin{aligned}
    {}{\mathcal D}^{\mathrm{LUE},sl}_0
&=\frac{d^{3}}{dy^{3}}
-4y\frac{d}{dy}
+2,\quad
{}{\mathcal D}^{\mathrm{LUE,sl}}_1
=-\tau_{\rm l}\bigg(3y\frac{d^3}{dy^3}
+4\frac{d^2}{dy^2}-4y^2\frac{d}{dy}-2y\bigg)
-4y^2\frac{d}{dy}
+4y,\\
{}{\mathcal D}^{\mathrm{LUE,sl}}_2
&=\tau_{\rm l}^2\bigg(3y^2\frac{d^3}{dy^3}
+8y\frac{d^2}{dy^2}+2\frac{d}{dy}\bigg)
+4{\tau}_{\rm l} y^3\frac{d}{dy},\quad
{}{\mathcal D}^{\mathrm{LUE,sl}}_3
=-\tau_{\rm l}^3\bigg(y^3\frac{d^3}{dy^3}
+4y^2\frac{d^2}{dy^2}
+2y\frac{d}{dy}\bigg).
\end{aligned}
\end{align}

We see that (\ref{4.24}) is identical in structure to (\ref{4.17r}), while (\ref{4.25}) is identical to (\ref{4.17X}) with ${\tau} \mapsto - {\tau}_{\rm l}$. 
Hence, defining $\{ r_j^{\rm L,sl}(y)\}$ in analogy with (\ref{4.18X}), the coupled differential equations of Proposition \ref{P4.3X} remain valid with this replacement.
Assuming $r_j^{\rm L,sl}(y)$ for $j=1,2$ are determined as the particular solutions with no additive contribution from $r_0^{\rm L}(y)$
(as is the case for the right soft edge), it follows that
these functions have the form (\ref{rjLUEsr}) with coefficients as in
Table \ref{T3}, now setting $\tau = - {\tau}_{\rm l}$.

\subsection{The hard edge} \label{S4.4}
Generally in the Laguerre ensembles, the hard edge refers to the neighbourhood of the origin with the parameter $a$ fixed. In the case of the density with $\beta$ even, and the probability density function of the smallest eigenvalue for general $\beta > 0$ with $a$ even,
or for $\beta = 2$ with general $a > -1$, it was demonstrated in \cite{FT19} (in the latter setting, see also
\cite{EGP16,Bo16,PS16,HHN16}) that an optimal rate of convergence is obtained through the hard edge scaling
\begin{equation}\label{4.15}
x = {y \over 4 N_{\rm h}'}, \qquad N_{\rm h}':= N + {a \over \beta}.
\end{equation}
Our interest here is specific to the LUE, which corresponds to $\beta = 2$.
We know from \cite[Eqns.~(2.15)--(2.17)]{FT19} that the hard edge scaled density $\rho_{(1),N}^{\rm LUE,h}(y) = (1 / 4 N_{\rm h}')\rho_{(1),N}^{\rm LUE,h}(y/4 N_{\rm h}')$ has the large $N$ expansion
\begin{equation}\label{4.16}
\rho_{(1),N}^{\rm LUE,h}(y) = r_0^{\rm L,h}(y) + {1 \over (N_{\rm h}')^2} r_1^{\rm L,h}(y) + {\rm O} \bigg ( {1 \over (N_{\rm h}')^3} \bigg ),
\end{equation}
where $ r_0^{\rm L,h}(y), r_1^{\rm L,h}(y)$ are given in terms of Bessel functions according to 
\begin{align}\label{4.17}
& r_0^{\rm L,h}(y) = {1 \over 4} \Big ( (J_a(\sqrt{y}))^2 - J_{a+1}(\sqrt{y})
J_{a-1}(\sqrt{y}) \Big ), \nonumber \\
& r_1^{\rm L,h}(y) = - {1 \over 192} \Big (
(2 y + a^2) (J_a(\sqrt{y}))^2 + 4 \sqrt{y} J_a(\sqrt{y}) 
J_a'(\sqrt{y}) + y (J_a'(\sqrt{y}))^2 \Big ).
\end{align}
Here, the notation $J_a'(\sqrt{y})$ is for $J_a'(u)|_{u = \sqrt{y}}$.
We remark too that the limit law as specified by $r_0^{\rm L,h}(y)$ was already known from the  earlier work \cite{Fo93}.

Our general method of analysing the linear differential equation satisfied by the density allows us to easily demonstrate that the expansion
(\ref{4.16}) contains only even powers of $(N_{\rm h}')^{-1}$.

\begin{proposition}\label{P4.5}
With $\mathcal D^{\rm LUE}$ specified by (\ref{4.0}), we have
\begin{equation}\label{4.18}
\mathcal D^{\rm LUE} \Big |_{x = y/(4 N_{\rm h}')} =
y^3 {d^3 \over d y^3} +4 y^2 {d^2 \over d y^2} +( y-a^2+2) y {d \over d y} + {y \over 2}-a^2 - {y^3 \over (4 N_{\rm h}')^2} {d \over d y}.
\end{equation}
Hence, the hard edge density permits
the large $N_{\rm h}'$ expansion
\begin{equation}\label{4.4h}
\rho_{(1),N}^{\rm LUE,h}(y) = r_0^{\rm L,h}(y) + (N_{\rm h}')^{-2} r_1^{\rm L,h}(y) +
(N_{\rm h}')^{-4} r_2^{\rm L,h}(y) + \cdots
\end{equation}
(cf.~(\ref{4.16})).
 Further, the successive terms are related by the nested inhomogeous differential equations
\begin{equation}\label{4.4j}
\mathcal D_0^{\rm LUE,h} r_j^{\rm L,h}(y) = - \mathcal D_1^{\rm LUE,h} r_{j-1}^{\rm L,h}(y), \quad r_{-1}^{\rm L,h}(y) = 0 \: \: (j=0,1,2,\dots ),
\end{equation}
 where we have used the notation $\mathcal D_0^{\rm LUE,h} +  (N_{\rm h}')^{-2} \mathcal D_1^{\rm LUE,h}$ for the right-hand side of (\ref{4.18}).
\end{proposition}

\begin{proof}
    The form (\ref{4.18}) is immediate from (\ref{4.0}).
  The existence of a well defined large $N$ expansion of the hard edge scaled density in inverse powers of $N$ is assured by (\ref{E6}) and the working of
Appendix D below. 
    Since the operator (\ref{4.18}) acting on the hard edge density
$\rho_{(1),N}^{\rm LUE,h}(y)$ gives zero, 
(\ref{4.4h}) follows.
Substituting this in the  rewritten differential equation gives (\ref{4.4j}).
    \end{proof}

    Using the Bessel function identities
\begin{equation}\label{4.18t}    
{2 a \over u} J_a(u) = J_{a-1}(u) + J_{a+1}(u), \quad
2 J_a'(u) = J_{a-1}(u) - J_{a+1}(u),
\end{equation}
we see that the $r_0^{\rm L,h}(y)$ as specified in (\ref{4.17}) can be rewritten as
\begin{equation}\label{4.19} 
r_0^{\rm L,h}(y) = {1 \over 4} \bigg (
\Big ( 1 - {a^2 \over y} \Big ) (J_a(\sqrt{y}))^2 + (J_a'(\sqrt{y}))^2 \bigg ).
\end{equation}
Comparing this with the functional form of $r_1^{\rm L,h}(y)$  in (\ref{4.17}) suggests that we seek solutions of the coupled equations
(\ref{4.4j}) of the form
\begin{equation}\label{2.RL}
 r_j^{\rm L, h}(y) = \alpha_j^{\rm L,h}(y) {1 \over y}  ( J_a(\sqrt{y})  )^2 +
  \beta_j^{\rm L,h}(y)  ( J_a'(\sqrt{y})  )^2  + \gamma_j^{\rm L,h}(y)
  {1 \over \sqrt{y}}  J_a(\sqrt{y})  J_a'(\sqrt{y}),
\end{equation} 
where $ \alpha_j^{\rm L,h}(y), \beta_j^{\rm L,h}(y),   \gamma_j^{\rm L,h}(y)$ are polynomials; cf.~(\ref{2.R}). The explicit forms of the latter can be read off from (\ref{4.19}) and the second equation in (\ref{4.17}); we record their values in Table \ref{T3}.

\begin{table}[h!]
\centering
\begin{tabular}{c|c|c|c} 
$j$ & $\alpha_j^{\rm L,h}(y)$ & $\beta_j^{\rm L,h}(y)$ & $\gamma_j^{\rm L,h}(y)$ \\
\hline
  0  &  ${1 \over 4} (y - a^2)$ & ${1 \over 4}$  & 0 \\[.2cm]
  1   & $- {1 \over 192} (2y + a^2)y$ & $-{1 \over 192}y$  & $-{1 \over 48}y$
  \end{tabular}
  \caption{Explicit forms of some low order cases of the coefficients in
  (\ref{2.RL}).}
  \label{T4}
  \end{table}

  The functional form (\ref{2.RL}) is compatible with the nested differential equations (\ref{4.4j}) since it remains closed under the derivative operation $y {d \over dy}$, or equivalently $\mathcal D_y:= {d \over d y} y$. The relevance of this property, to now be verified, is
  that the differential operator (\ref{4.18}) can be written as a polynomial of either of these operations.

  \begin{proposition}\label{P4.3}
Let $\mathcal D_y$ be as defined above. Let
\begin{equation}\label{4.20}
b_1(y) = {1 \over y}  ( J_a(\sqrt{y})  )^2, \quad
b_2(y) =  ( J_a'(\sqrt{y})  )^2, \quad
b_3(y) = {1 \over \sqrt{y}}   J_a(\sqrt{y})  J_a'(\sqrt{y}). 
\end{equation} 
We have
\begin{equation}\label{4.21}
\mathcal D_y b_1(y) = b_3(y), \quad
\mathcal D_y b_2(y) = - (y - a^2) b_3(y), \quad \mathcal D_y b_3(y) = 
{1 \over 2}(- (y - a^2)  b_1(y) +  b_2(y)).
\end{equation} 
\end{proposition}

\begin{proof}
These follow from the definitions and the rules for differentiation, making use too of the second order differential equation satisfied by the Bessel function, written in the form
\begin{equation}\label{4.22}
J_a''(u) = - {1 \over u^2} \Big ( u J_a'(u) + (u^2 - a^2) J_a(u) \Big ).
\end{equation}
\end{proof}

\begin{corollary} \label{Co4.7}
In addition to the differential relation (\ref{4.4j}) with $j=1$, we have
the explicit differential relation
\begin{equation}\label{4.22a}
r_1^{\rm L,h}(y) = - \Big ( {1 \over 12} \mathcal D_y^2 r_0^{\rm L,h}(y) +
{1 \over 48}  \mathcal D_y(y r_0^{\rm L,h}(y)) + {1 \over 24}  ( a^2 - 2)  \mathcal D_y r_0^{\rm L,h}(y) \Big ).
\end{equation} 
Let the generating function of there being $k$ ($k=0,1,2,\dots$) eigenvalues in the interval $(0,s)$ in the LUE be expanded for large $N_{\rm h}'$ according to
\begin{equation}\label{1.11j}
\mathcal E_{N,2}((0,y/(4N_{\rm h}'));x^a e^{-x};\xi)   = 
 \mathcal E_{2}^{\rm h}((0,y);\xi)  +
 {1 \over (N_{\rm h}')^2}  \mathcal E_{2}^{1,\rm L, h}((0,y);\xi)  + {1 \over (N_{\rm h}')^4}  \mathcal E_{2}^{2,\rm L, h}((0,y);\xi)  + \cdots;
 \end{equation}
 existence of an expansion of this form follows from Proposition \ref{PE2} below. Under the assumption
 that $ \mathcal E_{2}^{\rm h}((0,y);\xi)$ and
 $\mathcal E_{2}^{1,\rm L, h}((0,y);\xi)$ are related by a 
 $\xi$-independent second order differential operator applied to the former, we have
 \begin{equation}\label{1.11k}
   \mathcal E_{2}^{1,\rm L, h}((0,y);\xi) = - 
   \Big ( {1 \over 12}  y^2{d \over d y} + {1 \over 48}  y^2 + {1 \over 24}  a^2 y \Big ) {d \over d y}
    \mathcal E_{2}^{\rm h}((0,y);\xi).
  \end{equation}  
\end{corollary}

\begin{proof}
Using Proposition \ref{P4.3} and the form of $r_0^{\rm L,h}(y)$ as implied by Table \ref{T3}, we compute
 \begin{equation}\label{1.11m}
 \mathcal D_y r_0^{\rm L,h}(y) = {1 \over 4}  ( J_a(\sqrt{y})  )^2, \quad
 \mathcal D_y^2 r_0^{\rm L,h}(y) =  {1 \over 4} \Big ( \sqrt{y} J_a(\sqrt{y})
 J_a'(\sqrt{y}) + ( J_a(\sqrt{y}))^2 \Big ).
 \end{equation} 
 These allow (\ref{4.22a}) to be verified. The formula (\ref{1.11k}) then
 follows from the same argument as that leading to (\ref{1.17}).
\end{proof}

\begin{remark} \label{R4.8}
${}$ \\
1.~From \cite{Fo93} we have the simple exact result $\mathcal E_{N,\beta}((0,y/(4N));e^{-\beta x/2};\xi) |_{\xi = 1} = e^{-\beta y / 8}$. Thus, all correction terms in a large $N$ expansion are identically zero. Indeed, substituting 
\begin{equation}\label{5.37f}
\mathcal E_{2}^{\rm h}((0,y);\xi) |_{\xi = 1} = e^{- y / 4}
\end{equation}
in the right-hand side of (\ref{1.11k}) with $a=0$ (in keeping with the weight being $e^{-\beta x/2}|_{\beta=2}$) returns zero. Another simple exact result 
is that \cite[Prop.~10]{FT19}
\begin{multline}\label{5.37}
 \mathcal   E_{N,\beta}((0,y/(4N_{\rm h}'));x e^{-\beta x/2};\xi) |_{\xi = 1} =
e^{-\beta y / 8}  \, {}_0F_1 \Big ( {2 \over \beta}; {y \over 4} \Big ) + \\
{1 \over (N_{\rm h}')^2} {y \over 48} e^{-\beta y/8} \bigg (
-(1-1/\beta) \, {}_0F_1 \Big ( {2 \over \beta}; {y \over 4} \Big ) +
((1-1/\beta) + y \beta/8) \, 
 {}_0F_1 \Big ( {2 \over \beta}+1; {y \over 4} \Big ) \bigg ) + \cdots .
 \end{multline}
 In the case $\beta = 2$, we can check that the functional forms of the leading term and the first correction are consistent with (\ref{1.11k}) when setting $a=1$ (here, the weight is $xe^{-x}$). \\
 2.~Characterisations of both 
 $\mathcal E_{2}^{\rm h}((0,y);\xi)$ and
 $\mathcal E_{2}^{1,\rm L, h}((0,y);\xi)$ in terms of a particular $\sigma$-Painlev\'e III$'$ transcendent are known
 \cite{TW94,FW02,FT19}. This allows us to verify 
 (\ref{1.11k}), and thus establish its validity without any assumptions. The required working is sketched in Appendix C.
 \\
 3.~The generating function  $\mathcal E_{2}^{\rm h}((0,y);\xi)$ with $\xi = 1$ and the Laguerre parameter $a$ a positive integer permits an expression as a particular matrix average over Haar distributed unitary matrices of size $a$ (see, e.g., \cite[Eq.~(2.1)]{FW26}). This matrix in turn has an interpretation in terms of the length of the longest increasing subsequence of a random permutation \cite{Ra98}. In this context, the scaled $a \to \infty$ limit of $\mathcal E_{2}^{\rm h}((0,y)$ with $y \mapsto a^2 - 2 a (a/2)^{1/3} y$ is of interest \cite{BF03}. 
 The optimal expansion parameter and the resulting functional forms in the corresponding asymptotic expansion
 are given in \cite{BJ13,FM23,Bo24x,Bo24y}. The former for $\beta =2$ is $a^{-2/3}$. From the present viewpoint,
 this can be understood from  the operator decomposition 
\begin{equation}
{\mathcal D}^{\mathrm{LUE,h}} \Big |_{y \mapsto a^2 - 2 a (a/2)^{1/3} y} =
 a^{2}\,{}{\mathcal D}^{\mathrm{LUE,hs}}_{0}
+a^{4/3}\,{}{\mathcal D}^{\mathrm{LUE,hs}}_{1}
+a^{2/3}\,{}{\mathcal D}^{\mathrm{LUE,hs}}_{2}
+{}{\mathcal D}^{\mathrm{LUE,hs}}_{3}.
\end{equation}
Here, ${\mathcal D}^{\mathrm{LUE,h}}$ is the differential operator specified in Proposition \ref{P4.5}. For economy of space, we do not list the explicit forms of ${\mathcal D}^{\mathrm{LUE,hs}}_{i}$,
except to say that for $i=0$, this is equal to the operator on the left-hand side of (\ref{2.1}).
 \end{remark}

Not available in earlier literature is the evaluation of $r_2^{\rm L,h}(y)$. We would like to compute this here. The starting point is to use the form of the correlation kernel $K_N^{\rm LUE}(x,y)$ for the LUE specified in \cite[Eq.~(2.4)]{FT19}. The LUE density is the limit $x \to y$ of this. Making use of L'H\^{o}pital's rule and a Laguerre polynomial identity, we obtain
 \begin{align}\label{E6}
 \rho_{(1),N}^{\rm LUE}(x) & = {N! \over \Gamma(a+N)} x^a e^{-x} \Big (
 L_{N-1}^{(a)}(x)  L_{N}^{(a)}(x) -  L_{N-1}^{(a+1)}(x)  L_{N}^{(a-1)}(x) \Big ) \nonumber \\
 & = {\Gamma(N+a+1) \over (\Gamma(a+1))^2 \Gamma(N)} x^a e^{-x} \Big (
  {}_1 F_1 (-N+1;a+1;x)  {}_1 F_1 (-N;a+1;x) \nonumber \\
  &\qquad \qquad -{a \over a + 1}
   {}_1 F_1 (-N+1;a+2;x) {}_1 F_1 (-N;a;x) \Big ),
  \end{align}
  where the second equality makes use of the hypergeometric polynomial form of Laguerre polynomials given in (\ref{E1}) below.
  In relation to computing $r_2^{\rm L,h}(y)$, we have from (\ref{4.15}) with $\beta = 2$ and (\ref{4.16}) that we require the
  large $N$ expansion of (\ref{E6}) with
  $x = y/(4 N_{\rm h}') $. From Proposition \ref{PE1} of Appendix D, we have available the large
  $N$ form of the hypergeometric polynomials in (\ref{E6}) with $ x \mapsto x/N$. An appropriate $N$-dependent scaling of $x$ allows for the latter to be  used in (\ref{E6}),
  leading then to the sought evaluation.

  \begin{proposition}\label{P4.7}
  In the notation of (\ref{2.RL}), $r_2^{\rm L,h}(y)$ is specified by
  \begin{align}\label{E7i}
  &\alpha_2^{\rm L,h}(y) = {y \over 92 160} \Big ( 14 a^2 (-4+a^2) + 2 (48-9a^2) y - 26 y^2 \Big ), \nonumber \\
  &\beta_2^{\rm L,h}(y) = {y \over 92 160} \Big ( 192-128a^2+20a^4 +
  (-104+20a^2) y + 5 y^2 \Big ),  \nonumber \\
   &\gamma_2^{\rm L,h}(y) = {y \over 92 160} \Big ( 14 (-4 + a^2) + 6 y \Big ).
   \end{align}
  \end{proposition}

  \begin{proof}
According to (\ref{E6}) and (\ref{E1}), we require use of Proposition \ref{PE1} with $\alpha = a, a \pm 1$. This is immediate, since the operators in (\ref{E2}) and (\ref{E2a}) are independent of $\alpha$. 
In the latter formulas it is convenient to make use of the derivative formula
 \begin{equation}\label{E6i}
 {\sf D}_x^k \, {}_0 F_1 (\alpha+1;-x) = {(-x)^k \over \prod_{l=1}^k (l + a)}
 \, {}_0 F_1 (\alpha + k + 1;-x),
  \end{equation}
  which is implemented by default in the computer algebra system Mathematica.
  With $x \mapsto x/N$, we use the results of Proposition \ref{PE1} to expand the difference of hypergeometric polynomials in 
  (\ref{E6}) in powers of $N^{-1}$. After this,
we must replace $x$ by ${y \over 4} (1 - \epsilon)$, $\epsilon := a/(2 N_{\rm h}')$,
and now expand in powers of $1/N_{\rm h}'$. In theory, this is carried out using a 
Taylor series, 
 \begin{equation}\label{E6a}
f((1-\epsilon)y/4) = f(y/4) - \epsilon {\sf D}_x f(x) \Big |_{x \mapsto y/4} + {\epsilon^2 \over 2!} {\sf D}_x^2 f(x) \Big |_{x \mapsto y/4} + \cdots.
 \end{equation}
In practice, this step can be carried out internally in Mathematica, which automatically evaluates the
derivatives according to (\ref{E6i}), as previously remarked.
Proceeding in this way, we arrive at the expansion in powers of $1/N_{\rm h}'$
of the combination of hypergeometric polynomials in (\ref{E6}) with $x$ replaced as specified. Again internally in Mathematica, we update this expansion by multiplying it by the expansion in powers of $1/N_{\rm h}'$ of $e^{-x/(4N_{\rm h}')}$. One observes that an expansion in terms of $1/(N_{\rm h}')^2$ results.

With regards  to the prefactor $\Gamma(N+a+1)/\Gamma(N)$, we first assume $a \in \mathbb Z_{\ge 0}$.  Under this circumstance,
\begin{equation}\label{EZ}
{\Gamma(N+a+1) \over \Gamma(N)} = N^{a+1} \prod_{l=0}^a \Big ( 1 + {l \over N} \Big ).
\end{equation}
Making use of (\ref{E3}), terms in the $1/N$ expansion are seen to be polynomials in $a$, thus allowing us to remove the restriction
$a \in \mathbb Z_{\ge 0}$ (see also \cite{TE51}). This latter expansion, multiplied by $(4N_{\rm h}')^{-(a+1)}$, then expanded in powers of $1/N_{\rm h}'$, shows
\begin{multline}\label{EZ1}
(4N_{\rm h}')^{-(a+1)} {\Gamma(N+a+1) \over \Gamma(N)} \\ = 4^{-(a+1)} \Big (
1 - {a^3 + 3 a^2 + 2 a \over 24 (N_{\rm h}')^2} + { 5 a^6 + 12 a^5 - 25 a^4 -60 a^3 + 20 a^2 + 48 a \over 5760 (N_{\rm h}')^4} + \cdots \Big ).
\end{multline}
(The working in the proof of Proposition \ref{PE2}  below assures that this expansion is in powers of $(N_{\rm h}')^{-2}$.)

Multiplying (\ref{EZ1}) by the expansion obtained in the paragraph before, and multiplying by the factor of $1/(\Gamma(a+1))^2$ as also appears in (\ref{E6}), then using the $\tt FullSimplify$ command in Mathematica, we obtain a form equivalent to (\ref{2.RL}) (specifically, it involves $J_{a-1}(\sqrt{y})$ instead of $J_{a}'(\sqrt{y})$, which can be eliminated by appropriate use of (\ref{4.18t}))
with the coefficients as specified in (\ref{E7i}).
\end{proof} 

\begin{remark} ${}$ \\
1.~Above (\ref{E6}) mention is made of the LUE correlation kernel
$K_N^{\rm LUE}(x,y)$. In subsection D.2 of Appendix D we make of its form in terms of hypergeometric polynomials. Moreover, we use that form to exhibit that the hard edge scaled large $N_{\rm h}'$ expansion is of the form
 \begin{equation}\label{E1X}
 \frac{1}{4 N_{\rm h}'} K_N^{\rm LUE}(x/4N_{\rm h}',y/4N_{\rm h}') = K_{\infty,0}^{\rm h}(x,y) +
 {1 \over (N_{\rm h}')^2} K_{\infty,1}^{\rm h}(x,y) +
  {1 \over (N_{\rm h}')^4} K_{\infty,2}^{\rm h}(x,y) + \cdots
\end{equation}
Since the correlation kernel fully determines the general $k$-point correlation functions, which in turn fully determines the generating function in (\ref{1.11j}) (see \cite[Eq.(9.1)]{Fo10}), this establishes the
validity of the expansion in (\ref{1.11j}). An analogous argument has been used in the recent work \cite{FS25} in relation to establishing an expansion in powers of $1/N^2$ for the spacing distributions of the bulk scaled circular ensembles.\\
2.~From the viewpoint of the inhomogeneous equation (\ref{4.4j}), the solution for $j=2$ given in Proposition \ref{P4.7} does not contain an additive multiple of $r_0^{\rm L,h}(y)$, and so can be characterised as the particular solution with this property. Moreover, a consistency check on Proposition \ref{P4.7} is that (\ref{4.4j}) with $j=2$ is indeed satisfied. This can be verified with the aid of computer algebra.

\end{remark}

\section{The LOE and LSE edge densities} \label{S5}

Denote by $\rho_{(1), \beta, N}^{{\rm L} \beta {\rm E}}(x)$ the density of the Laguerre $\beta$ ensemble; for even $\beta$, it is known from \cite{RF21} to satisfy a differential equation:
\begin{align}\label{5.1}
    \mathcal{D}_{N}^{\rm L\beta E} \rho_{(1), \beta, N}^{{\rm L} \beta {\rm E}}(x)=0,
\end{align}
where $\mathcal{D}_{N}^{\rm L\beta E}$ is a linear differential operator of order $\beta+1$. 
By a duality satisfied by the  spectral moments \cite{DE06,FRW17},
the same equation characterises the density with $\beta\to 4/\beta$,
up to mappings of $N$ and $a$.
Explicit forms of the operators are given for $\beta=1,2,4$ in \cite{RF21}; recall (\ref{4.0}) for the case of $\beta = 2$ (LUE) and see \cite[Eq.~(2.29)]{RF21}, with the proviso that the Laguerre weight is taken as $x^a e^{-x}$ instead of as
$x^a e^{-\beta x/2}$ as in (\ref{1.2}), for $\beta = 1,4$.
We now study the cases $\beta = 1$ (LOE) and
$\beta = 4$ (LSE). In the interest of efficiency of presentation, attention will be focused on two classes of edge setting (for the LUE, we considered four): the right soft edge with $a={\beta \over 2}(\gamma-1)N+({\beta \over 2}-1)$, as is relevant to Wishart matrices in multidimensional statistics
and considered in the present context in \cite{Bo24,Bo25}, and the
hard edge.

\subsection{Right soft edge --- the case \texorpdfstring{$a={\beta \over 2}(\gamma -1)N+({\beta \over 2}-1)$}{a=β(γ-1)N/2+(β/2-1)}}\label{S5.1}

Let $\hat{N}_{\rm s}'$ be as in \eqref{Nh}.
For this choice of $a$, changing variables in 
the differential operators \cite[Eq.~(2.29) with $x \mapsto \beta x / 2$]{RF21}
according to (\ref{vJ}) to give a soft edge scaling gives the form
\begin{equation}\label{5.9}
\Bigg(\sum_{k=0}^5
\frac{1}{(\hat{N}_{\rm s}')^{2k/3}}{\mathcal D}^{\mathrm{L},\beta,sr}_k
\Bigg)\rho_{(1),N}^{\mathrm{L\beta E},sr}(y)=0.
\end{equation}
The operators ${\mathcal D}^{\mathrm{L},\beta,sr}_k$,
determined with the aid of computer algebra, are listed as a proposition in the following.

\begin{proposition}
With reference to (\ref{5.9}), and with $\tau$ as in \eqref{Nh}, we have
\begin{align}
{\mathcal D}^{\mathrm{L},\beta,sr}_0&=\frac{4}{\beta }\frac{d^5}{dy^5}-20 y \frac{d^3}{dy^3}+12 \frac{d^2}{dy^2}+16 \beta  y^2\frac{d}{dy}-8 \beta  y, \label{5.3}
\\ {\mathcal D}^{\mathrm{L},\beta,sr}_1
&=\frac{20\tau}{\beta }\,y\frac{d^5}{dy^5}
+\frac{40\tau}{\beta }\frac{d^4}{dy^4}
-5(4+12\tau)\,y^2\frac{d^3}{dy^3}
+(24-52\tau)\,y\frac{d^2}{dy^2}.
 \nonumber
   \\&\quad
+16\beta(2+\tau)\,y^3\frac{d}{dy}
+2(16\tau-12)\frac{d}{dy}
-8\beta(3-\tau)\,y^2
,
\\ {\mathcal D}^{\mathrm{L},\beta,sr}_2&=
\frac{40 \tau^2 y^2}{\beta }\frac{d^5}{dy^5}
+\frac{160 \tau^2 y }{\beta }\frac{d^4}{dy^4}
+\left(\frac{93 \tau^2}{\beta }-60 \tau y^3-60 \tau^2 y^3\right) \frac{d^3}{dy^3} \nonumber\\
&\quad-\left(16 \tau +140 \tau^2\right) y^2 \frac{d^2}{dy^2}
+\left(16\beta y^3+32\beta \tau y^3-8\tau\right) y \frac{d}{dy} \nonumber\\
&\quad-\left(16\beta y^3-8\beta \tau y^3+16\tau-10\tau^2\right)
,
    \end{align}
    \begin{align}
 {\mathcal D}^{\mathrm{L},\beta,sr}_3&=\frac{40 \tau^3 y^3}{\beta }\frac{d^5}{dy^5}
+\frac{240 \tau^3 y^2}{\beta}\frac{d^4}{dy^4}
+\left(\frac{279 \tau^3}{\beta }-60 \tau^2 y^3-20 \tau^3 y^3\right) y \frac{d^3}{dy^3} \nonumber\\
&\quad+\left(\frac{38 \tau^3}{\beta }-104 \tau^2 y^3-76 \tau^3 y^3\right)\frac{d^2}{dy^2}
+\left(16\beta \tau y^3-8\tau^2-32\tau^3\right) y^2 \frac{d}{dy} \nonumber\\
&\quad-\left(12\tau^2-2\tau^3\right) y,
\\ {\mathcal D}^{\mathrm{L},\beta,sr}_4&=\frac{20 \tau^4 y^4}{\beta }\frac{d^5}{dy^5}
+\frac{160 \tau^4 y^3}{\beta}\frac{d^4}{dy^4}
+\left(\frac{279 \tau^4}{\beta }-20 \tau^3 y^3\right) y^2 \frac{d^3}{dy^3} \nonumber\\
&\quad+\left(\frac{76 \tau^4}{\beta }-64 \tau^3 y^3\right) y \frac{d^2}{dy^2}
-\left(\frac{\tau^4}{\beta }+24 \tau^3 y^3\right) \frac{d}{dy}
-4 \tau^3 y^2,
\\ {\mathcal D}^{\mathrm{L},\beta,sr}_5&=\tau^5\bigg(\frac{4 y^5}{\beta }\frac{d^5}{dy^5}+\frac{40 y^4 }{\beta}\frac{d^4}{dy^4}+\frac{93 y^3}{\beta}\frac{d^3}{dy^3}+\frac{38 y^2}{\beta}\frac{d^2}{dy^2}-\frac{y}{\beta
   }\frac{d}{dy}+\frac{1}{\beta }\bigg).\label{5.8}
\end{align}
\end{proposition}

Comparing (\ref{5.3}) with the first operator in (\ref{3.11})
shows that they are equal: ${\mathcal D}^{\mathrm{L},\beta,sr}_0 = \mathcal D_0^{\rm G, \beta}$. This is in keeping with the limiting soft edge density being the same for the Gaussian and Laguerre ensembles, depending only on the symmetry class ($\beta$ value).
A feature of the explicit forms of ${\mathcal D}^{\mathrm{L},\beta,sr}_k$ for $k = 1,2$ is
 that for ${\tau} \to 0$, $ {\mathcal D}^{\mathrm{L},\beta,sr}_k \to {\mathcal D}^{\mathrm{G},\beta}_k$,
where the ${\mathcal D}^{\mathrm{G},\beta}_k$ are as specified in Proposition \ref{P3.2} (${\mathcal D}^{\mathrm{G},\beta}_k=0$ for $k>2$).
This property is consistent with \cite[Remark 4.1]{Bo25}.

One observes too that analogous to a property of the differential operators (\ref{3.11}), we have that the scaled operators
\begin{equation}\label{3.11L}
\tilde{\mathcal D}_k^{\rm L,\beta} := \beta^{(k-2)/3} 
{\mathcal D}_k^{{\rm L}, \beta, sr} \Big |_{y \mapsto \beta^{-1/3} y}
\end{equation}
are independent of $\beta$. This suggests proceeding analogous to (\ref{2.8a}) and introducing the scaled expansion
\begin{equation}\label{2.8aL} 
 \beta^{\frac{1}{6}}\rho_{(1),N}^{\mathrm{L},\beta,sr}(y) =
 {}{r}_0^{\rm L,\beta,sr}(\beta ^{\frac{1}{3}}y) + (\sqrt{\beta}\hat{N}'_{\rm s})^{-2/3} {}{r}_1^{\rm L,\beta,sr}(\beta ^{\frac{1}{3}}y) +
(\sqrt{\beta}\hat{N}'_{\rm s})^{-4/3} {}{r}_2^{\rm L,\beta,sr}(\beta ^{\frac{1}{3}}y) + \cdots.
 \end{equation}
From this expansion, the analogue of (\ref{3.14E}) in Corollary \ref{C3.3} can be formulated.

 \begin{corollary}
     For $j=0,1,\dots$ we have
\begin{equation}\label{2.8aL1} 
 \sum_{k=0}^5    \tilde{\mathcal D}_k^{\rm L,\beta} \, {r}_{j-k}^{\rm L,\beta,sr}(y) = 0,
  \end{equation}    
  where ${r}_{l}^{\rm L,\beta,sr}(y) = 0$ for $l < 0$.
 \end{corollary}

As is consistent with the fact that $\tilde{D}_0^{\rm L, \beta} =
\tilde{D}_0^{\rm G, \beta}$, one has for both $\beta = 1$ and
$4$ that ${r}_{0}^{\rm L,\beta,sr}(y) =
{r}_{0}^{\rm G,\beta}(y)$, with the latter as specified  by (\ref{rGbeta}) and Table \ref{T2}.
Moreover, we know from \cite{Bo25} that the
analogue to (\ref{rGbeta}) holds:
\begin{align}\label{rGbetaL}
      \begin{aligned}
          {}{r}_j^{{\rm L},\beta,sr}(y) =& {}{\alpha}_j^{\rm L,\beta,sr}(y)  ( {\rm Ai}(y)  )^2 +
  {}{\beta}_j^{\rm L,\beta,sr}(y)  ( {\rm Ai}'(y)  )^2 + {}{\gamma}_j^{\rm L,\beta,sr}(y)
   {\rm Ai}(y)  {\rm Ai}'(y)\\
   &+{}{\xi}_j^{\rm L,\beta,sr}(y){\rm Ai}(y){\rm AI}_\nu(y)+{}{\eta}_j^{\rm L,\beta,sr}(y){\rm Ai}'(y){\rm AI}_\nu(y)
      \end{aligned}
 \end{align}
 for certain polynomial coefficients
 ${}{\alpha}_j^{\rm L,\beta,sr}(y),\dots,{}{\eta}_j^{\rm L,\beta,sr}(y)$. For $j=0,1,2$, these are specified explicitly in \cite[Th.~4.1]{Bo25}. For $j=0$, we read these polynomials off from the first row of Table \ref{T2}.
 For $j=1$, they are specified by
 \begin{align}\label{5.12a} 
&{}{\alpha}_1^{\rm L,\beta,sr}(y) = {2 \tau - 1 \over 2} y^2,
\quad {}{\beta}_1^{\rm L,\beta,sr}(y) = - {2(2 \tau - 1) \over 5} y,\quad {\gamma}_1^{\rm L,\beta,sr}(y) =  {3 - \tau  \over 10} , \: \:
\nonumber \\& {\xi}_1^{\rm L,\beta,sr}(y) = - {3  \tau + 1 \over 10} y, \quad {\eta}_j^{\rm L,\beta,sr}(y)= -  {2 \tau - 1 \over 10 } y^2.
 \end{align}
 Setting $\tau = 0$ reclaims the values of these coefficients
 for $\beta = 1,4$ as specified in Table \ref{T2};
 recall the discussion in the paragraph below (\ref{5.8}).
 They are lengthier again for $j=2$ and we refer to \cite{Bo25} for their explicit form. Results from \cite[\S 6]{Bo24} and \cite{Bo25} also relate ${r}_j^{{\rm L},\beta,sr}(y)$ for $j=1,2$ (and $j=3$ too) to ${r}_0^{{\rm L},\beta,sr}(y)$ via a differential operator. For example,
 \begin{equation}
{d \over dy} \bigg ( - \Big( {2 \tau - 1 \over 5} \Big ) y^2 + { \tau - 3 \over 5}{d \over dy} \bigg )
 {r}_0^{{\rm L},\beta,sr}(y) = {r}_1^{{\rm L},\beta,sr}(y),
 \end{equation}
 which for $\tau=0$ reduces to (\ref{2.3ar}).

\medskip

\subsection{The hard edge} \label{S5.2}
In this subsection, we consider the hard edge scaled density for the LOE/LSE. 
To present the results in a unified way, we adopt a convention similar to that used in \cite{Bo25} for the weight functions:
\begin{align}\label{5.14}
    w_{\beta}(x):=\left\{\begin{array}{cc}
    x^{\frac{a-1}{2}}e^{-x},     & \beta=1, \\
    x^{a+1}e^{-x},     & \beta=4,
    \end{array}\right.
\end{align}
where $a$ is fixed at the hard edge. Denote by $\rho_{(1),N}^{\rm L\beta E,h}(y)$ the corresponding density scaled in the hard edge variable \eqref{4.15}. 
Introduce too the modification of \eqref{4.15},
\begin{equation}\label{hN}
\hat{N}_{\rm h}' = \begin{cases}  N  + (a - 1)/2, & \beta = 1, \\ N + (a+1)/4, & \beta = 4. \end{cases}
\end{equation}
Then, as a consequence of 
knowledge of the explicit form of the operator 
$\mathcal{D}_{N}^{\rm L\beta E}$ in \eqref{5.1}
for $\beta =1$ and 4,
we have the following result for the expansion of
the density with respect to (\ref{hN}).

\begin{proposition}
Let
\begin{align}\label{LSE-para}
     \tilde{a}=a^2-1,
 \end{align}
 and specify $\hat{N}_{\rm h}'$ by \eqref{hN}.
 For the 
 differential operator $\mathcal{D}_{N}^{\rm L\beta E}$ in
 (\ref{5.1}) with $\beta=1,4$, after a transformation of the Laguerre weight function
 $x^a e^{-x}\to w_\beta(x)$ as specified by (\ref{5.14})
 and introduction of the hard edge scaled variable, we have
    \begin{align}
&\left.\mathcal{D}_{N}^{\rm L\beta E}\right|_{\substack{x=y/(4\hat{N}_h')\\ w\to w_{\beta}}}
=
\Bigg(
\tilde a^{2}-(4+3\tilde a)\,y+2y^{2}
+\Big(4y^2-4(-3+\tilde a)y-16-14\tilde a+\tilde a^{2}\Big)y\frac{d}{dy} \nonumber\\
&\quad
+\big(38y+16-22\tilde a\big)y^{2}\frac{d^{2}}{dy^{2}}
+\big(10y+88-5\tilde a\big)y^{3}\frac{d^{3}}{dy^{3}}
+40y^{4}\frac{d^{4}}{dy^{4}}
+4y^{5}\frac{d^{5}}{dy^{5}}
\Bigg) \nonumber\\[0.3em]
& \quad
+\frac{y^2}{(2\sqrt{\beta}\hat{N}_{\rm h}')^{2}}
\Bigg(
\tilde a-y
+2(-2+\tilde a-2y)\,y\frac{d}{dy}
-16y^{2}\frac{d^{2}}{dy^{2}}
-5y^{3}\frac{d^{3}}{dy^{3}}
\Bigg) \nonumber
+\frac{y^{5}}{(2\sqrt{\beta}\hat{N}_{\rm h}')^{4}}
\frac{d}{dy}\\
& \quad :=\mathcal{D}_{0}^{\rm L\beta E,h}+\mathcal{D}_{1}^{\rm L\beta E,h}(2\sqrt{\beta}\hat{N}_{\rm h}')^{-2}+\mathcal{D}_{2}^{\rm L\beta E,h}(2\sqrt{\beta}\hat{N}_{\rm h}')^{-4}.
\end{align}
Thus, the hard edge density for the LOE/LSE satisfies
\begin{align}
    \bigg(\mathcal{D}_{0}^{\rm L\beta E,h}+\mathcal{D}_{1}^{\rm L\beta E,h}(2\sqrt{\beta}\hat{N}_{\rm h}')^{-2}+\mathcal{D}_{2}^{\rm L\beta E,h}(2\sqrt{\beta}\hat{N}_{\rm h}')^{-4}\bigg)\rho_{(1),N}^{\rm L\beta E,h}( y)=0
\end{align}
and so  permits the large $\hat{N}_{\rm h}'$ expansion
\begin{align}\label{5.20v}
    \rho_{(1),N}^{\rm L\beta E,h}( y) = {r}_0^{\rm L,h,\beta}(y) + {1 \over (2\sqrt{\beta}\hat{N}_{\rm h}')^2} {r}_1^{\rm L,h,\beta}(y)+ {1 \over (2\sqrt{\beta}\hat{N}_{\rm h}')^4} {r}_2^{\rm L,h,\beta}(y) + \cdots .
\end{align}
With $r_{-1}^{\rm L,h,\beta}(y)= {r}_{-2}^{\rm L,h,\beta}(y)=0$,
the successive terms are related by
\begin{align}\label{5.18}
    &\mathcal{D}_{0}^{\rm L\beta E,h} {r}_j^{\rm L,h,\beta}(y)=-\mathcal{D}_{1}^{\rm L\beta E,h} {r}_{j-1}^{\rm L,h,\beta}(y)-\mathcal{D}_{2}^{\rm L\beta E,h} {r}_{j-2}^{\rm L,h.\beta}(y),\qquad (j=0,1,2,\dots).
\end{align}
\end{proposition}

We remark that the operator $\mathcal{D}_{0}^{\rm L\beta E,h}$ was reported earlier in \cite[Th.~4.2]{RF21} with $\kappa x = y$ and a modified meaning of $\tilde{a}$. It is known from \cite{FNH99}, \cite[Eq.~(7.110) after correction and (7.155)]{Fo10} that for $\beta = 1,4$,
\begin{align}\label{5.21}
    {r}_0^{\rm L,h,\beta}\left({y \over 2}\right)=\frac{1}{2}\ \Big (J_a (\sqrt{y})^2-J_{a +1}(\sqrt{y}) J_{a -1}(\sqrt{y})\Big )+\frac{J_{a}(\sqrt{y})}{2\sqrt{y}}{\rm JI}_a(\sqrt{y}) ,
\end{align}
where 
\begin{align}\label{5.22}
    {\rm JI}_a(y):= \begin{cases}
    \int_y^\infty J_a(t)dt, & \beta = 1, \\
    - \int_0^y J_a(t)dt, & \beta = 4. \end{cases}
\end{align}
One should compare the definition of this latter quantity
with that of ${\rm AI}_\nu(y)$
in (\ref{3.6}).
% One notes that the right-hand side of (\ref{5.21}) is independent of $\beta$; compare the $\beta$-independent definition of $ {\rm JI}_a(y)$ in (\ref{5.22}) with 
%  the $\beta$-dependent definition 
% of . There is nonetheless indirect $\beta$-dependence through the parameter $a$ due to our redefinition of the weight $w_\beta(x)$ given in (\ref{5.14}).

The functional form (\ref{5.21}) suggests introducing the further basis functions
\begin{align}
    b_4(y)=\frac{1}{\sqrt{y}}J_a(\sqrt{y}){\rm JI}_a(\sqrt{y}),\quad b_5(y)=J_a'(\sqrt{y}){\rm JI}_a(\sqrt{y})
\end{align}
(cf.~the basis functions on the second line of (\ref{rGbeta})),
for which the analogue of \eqref{4.21} is given by 
\begin{align} \label{5.25}
    \mathcal D_y b_4(y)=\frac{1}{2}(b_4(y)+b_5(y)-yb_1(y)),\quad \mathcal D_y b_5(y)=\frac{1}{2}(b_5(y)-(y-a^2)b_4(y) -y b_3(y)).
\end{align}

The facts that the span of $\{b_1(y),\dots,b_5(y) \}$ with
polynomial coefficients is closed under the operation
$\mathcal D_y=\frac{d}{dy}y$ and that ${r}_0^{\rm L,h,\beta}(y/2)$ is expressible as such a linear combination encourages us to seek solutions of (\ref{5.18}) of a form extending (\ref{2.RL}),
\begin{multline}\label{2.RLx}
r_j^{\rm L, h, \beta}\Big ({y \over 2} \Big ) 
\\ = \alpha_j^{\rm L,h,\beta}(y) b_1(y) +
  \beta_j^{\rm L,h,\beta}(y) b_2(y)   + \gamma_j^{\rm L,h,\beta}(y)
  b_3(y) + \xi_j^{\rm L,h,\beta}(y) b_4(y) +
  \eta_j^{\rm L,h,\beta}(y) b_5(y),
\end{multline} 
where $ \alpha_j^{\rm L,h,\beta}(y), \dots, \eta_j^{\rm L,h,\beta}(y)$
are polynomials, the same for both $\beta = 1,4$. 
Comparing with (\ref{5.21}) shows that for the case $j=0$,
\begin{equation}\label{sn}
{\alpha}_0^{\rm L,h,\beta}(y) = \frac{y-a^2}{2}, \: \:
{\beta}_0^{\rm L,h,\beta}(y) = {1 \over 2}, \; \:
{\gamma}_0^{\rm L,h,\beta}(y) = 0, \: \:
{\xi}_0^{\rm L,h,\beta}(y) = {1 \over 2}, \: \:
{\eta}_0^{\rm L,h,\beta}(y) = 0.
\end{equation}

With knowledge of (\ref{sn}), and aided by computer algebra, we can compute a particular solution of the fifth order inhomogeneous equation
\begin{equation}\label{sn1}
\mathcal{D}_{0}^{\rm L\beta E,h} Y(y)=-\mathcal{D}_{1}^{\rm L\beta E,h} {r}_{0}^{\rm L,h,\beta}(y)
\end{equation}
corresponding to the case $j=1$ of (\ref{5.18}), and we can further identify the solution of the homogeneous equation in this class.

\begin{proposition}
Solutions of (\ref{sn1}) with $Y(y/2)$ of the form
(\ref{2.RLx}) have the structure
\begin{equation}\label{sn2}
Y(y/2) = C r_0^{\rm L, h, \beta}(y/2) +
P(y/2)
\end{equation}
where $C$ is a constant and $P(y/2)$ is such that the coefficients in (\ref{2.RLx}) do not contain an additive multiple of (\ref{sn}). Moreover, 
\begin{multline}\label{sn3}
P(y/2) = -\Big ( \frac{(y-a^2)^2+a^4-2a^2}{16} \Big ) b_1(y) \\
-
  \frac{y-3a^2+3}{24} b_2(y) - {y \over 12} b_3(y) +
  {y \over 48}  b_4(y) - 
 \frac{y+2a^2-2}{48} b_5(y).
\end{multline}
\end{proposition}

In the previous cases considered, we have found that both the first and second corrections are solutions of the corresponding inhomogeneous equation with no contribution from the solution of the homogeneous equation. Specialising to $\beta = 1$ and making use of the known finite $N$ expression
\cite{AFNV00}, \cite[Eq.~(7.154)]{Fo10}, \cite{Bo25}
\begin{multline}\label{2.RLy}
{1 \over 2} \rho_{(1),N}^{\rm L \beta E,h} \Big ( {y \over 2} \Big ) \Big |_{\beta = 1} =
\rho_{(1),N-1}^{\rm L U E,h}(y) -  {1 \over 4 \hat{N}_{\rm h}'} {(N-1)! \over 4 \Gamma(a-1+N)} x^{(a-1)/2} e^{-x/2} L_{N-1}^{(a)}(x) \\
\times \bigg ( 2 \int_0^x L_{N-2}^{(a)}(u)
 u^{(a-1)/2} e^{-u/2} \, du -
 {\Gamma((N+1)/2) \Gamma(a-1+N) \over 2^{a/2-3/2} \Gamma(N)
 \Gamma((N+a)/2)} \bigg ) \bigg |_{x = y/(4\hat{N}_{\rm h}')},
 \end{multline}
 we will investigate the validity of that property in the present setting. Thus, our 
 task is to compute the large
 $\hat{N}_{\rm h}' = N + (a-1)/2$ expansion of (\ref{2.RLy}), making explicit the ${\rm O}((\hat{N}_{\rm h}')^{-2})$
 first order correction, which 
 subject to the structural hypothesis (\ref{2.RLx}) must be
 of the form (\ref{sn2}).

 \begin{proposition}\label{PEv}
Consider the expansion (\ref{5.20v}) for $\beta = 1$.
As with the leading functional form (\ref{5.21}), after replacing the argument $y$ by $y/2$, the first correction has the functional form (\ref{2.RLx}). 
With $ r_0^{\rm L, h, \beta}(y/2)$ specified by the coefficients (\ref{sn}), and $P(y/2)$ specified by (\ref{sn3}), we have that $ r_1^{\rm L, h, \beta}(y/2)$ is specified by the right-hand side of (\ref{sn2}) with
$C = (1-a^2)/4$.
Thus, from the viewpoint of the underlying inhomogeneous differential equation, there is a contribution from the solution of the homogeneous part.
Specifically,  $ r_1^{\rm L, h, \beta}(y/2)$ is given by
(\ref{2.RLx}) with coefficients
\begin{multline}\label{sn4}
{\alpha}_1^{\rm L,h,\beta}(y) = \frac{2y - y^2}{16}, \quad
{\beta}_1^{\rm L,h,\beta}(y) = -{y  \over 24}, \quad
{\gamma}_1^{\rm L,h,\beta}(y) = -{y  \over 12}, \\
{\xi}_1^{\rm L,h,\beta}(y) = {y - 6 a^2 + 6 \over 48}, \quad
{\eta}_1^{\rm L,h,\beta}(y) = -{y + 2 a^2 -2 \over 48}.
\end{multline}

 \end{proposition}

 \begin{proof}
 Starting from the explicit functional form (\ref{2.RLy}), this requires an asymptotic analysis similar to that used for the proof of Proposition \ref{P4.7}. While the latter was to one order higher, in the present problem we are faced with the task of simplifying varying integrals. For this, as is seen in our detailed working of Appendix E, essential use is made of the second order differential equation satisfied by $J_a(u)$ (this gives a mechanism explaining why the asymptotic expansion has the structure
 (\ref{2.RLx})).
 \end{proof}

 We conclude with an analogue of Corollary \ref{Co4.7}.
 \begin{corollary}
In addition to the differential relation (\ref{5.18}) with $j=1$, we have the explicit differential relation
\begin{equation} \label{5.33new}
r_1^{\rm L,h,\beta}\Big(\frac{y}{2}\Big) = -\frac{1}{12}\Big(8\mathcal{D}_y^2r_0^{\rm L,h,\beta}\Big(\frac{y}{2}\Big)+\mathcal{D}_y\Big(yr_0^{\rm L,h,\beta}\Big(\frac{y}{2}\Big)\Big)+2(a^2-5)\mathcal{D}_yr_0^{\rm L,h,\beta}\Big(\frac{y}{2}\Big)\Big).
\end{equation}
For $\beta=1,4$ and $w_\beta(x)$ specified by (\ref{5.14}), let the generating function of there being $k$ ($k=0,1,2,\ldots$) eigenvalues in the interval $(0,s)$ be expanded for large $\hat{N}_{\rm h}'$ according to
\begin{multline}
\mathcal E_{N,\beta}((0,y/(4\hat{N}_{\rm h}'));w_\beta(x);\xi)   = 
 \mathcal E_{\beta}^{\rm h}((0,y);\xi)  +
 {1 \over (2\sqrt{\beta}\hat{N}_{\rm h}')^2}  \mathcal E_{\beta}^{1,\rm L, h}((0,y);\xi) 
 \\+ {1 \over (2\sqrt{\beta}\hat{N}_{\rm h}')^4}  \mathcal E_{\beta}^{2,\rm L, h}((0,y);\xi)  + \cdots.
 \end{multline}
 Under the assumption that $\mathcal{E}_{\beta}^{\rm h}((0,y);\xi)$ and $\mathcal{E}_{\beta}^{\rm 1,L,h}((0,y);\xi)$ are related by a $\xi$-independent second order differential operator applied to the former, we have
 \begin{equation} \label{5.35}
\mathcal{E}_\beta^{\rm 1,L,h}((0,y);\xi)=-\frac{1}{6}\Big(4y^2\frac{d}{dy}+y^2+(a^2-1)y\Big)\frac{d}{dy}\mathcal{E}_\beta^{\rm h}((0,y);\xi).
 \end{equation}
 \end{corollary}
\begin{proof}
Using Proposition \ref{P4.3} and its extension (\ref{5.25}), along with the form of $r_0^{\rm L,h,\beta}(y/2)$ given in (\ref{5.21}), we compute
\begin{align*}
 \mathcal{D}_y r_0^{\rm L,h,\beta}\Big(\frac{y}{2}\Big) &= \frac{1}{4}\Big( (J_a(\sqrt{y}))^2 +\frac{1}{\sqrt{y}}J_a(\sqrt{y}){\rm JI}_a(\sqrt{y})+J_a'(\sqrt{y}){\rm JI}_a(\sqrt{y})\Big),
\\  \mathcal{D}_y^2 r_0^{\rm L,h,\beta}\Big(\frac{y}{2}\Big) &= \frac{1}{8}\Big( (J_a(\sqrt{y}))^2 + \sqrt{y}J_a(\sqrt{y})J_a'(\sqrt{y})+\frac{a^2+1-y}{\sqrt{y}}J_a(\sqrt{y}){\rm JI}_a(\sqrt{y})
\\&\qquad+2J_a'(\sqrt{y}){\rm JI}_a(\sqrt{y})\Big).
\end{align*}
Thus, we verify (\ref{5.33new}). Changing variables by a factor of $2$ produces a relation between $r_0^{\rm L,h,\beta}(y)$ and $r_1^{\rm L,h,\beta}(y)$, which gives (\ref{5.35}) upon applying the same reasoning as that leading to (\ref{1.17}).
\end{proof}

\begin{remark}
 Recall from Remark \ref{R4.8} the exact results (\ref{5.37f}) and (\ref{5.37}). With our weight $w_\beta(x)$ being specified by (\ref{5.14}), the first of these gives that when $\beta=1$ and $a=1$, or $\beta=4$ and $a=-1$, we have
$$
\mathcal E_{\beta}^{\rm h}((0,y);\xi)|_{\xi=1}=e^{-y/4},
\quad  \mathcal E_{\beta}^{j,\rm{L,h}}((0,y);\xi)|_{\xi=1}=0,\\ \: (j=1,2,\cdots).
$$
This is indeed consistent with (\ref{5.35}). On the other hand, setting $\beta=1$ in the second of these results shows that when $a=3$, we have
$$
\mathcal E_1^{\rm h}((0,y);\xi)|_{\xi=1} = e^{-y/4}{}_0F_1\Big(2;\frac{y}{2}\Big),
\quad \mathcal E_1^{\rm 1,L,h}((0,y);\xi)|_{\xi=1} = \frac{y^2}{24}e^{-y/4}{}_0F_1\Big(3;\frac{y}{2}\Big).
$$
Calculation shows that this is again consistent with (\ref{5.35}). A similar check with $\beta=4$ and $a=0$ also holds.
\end{remark}

\section*{Acknowledgements}
The work of PJF is supported by a grant from the Australian Research Council, Discovery Project
DP250102552. AAR is supported by Hong Kong RGC grants GRF 16304724 and GRF 17304225.
 BJS is supported by  the Shanghai Jiao Tong
University Overseas Joint Postdoctoral Fellowship Program.
Correspondence on the topic of this work from F.~Bornemann is appreciated.

\appendix
\section*{Appendix A: An orthogonal polynomial approach to the third order differential equations}
\renewcommand{\thesection}{A} 
\renewcommand{\theHsection}{A} 
\setcounter{equation}{0}
\setcounter{theorem}{0}

Here, we consider the classical unitary ensembles, with eigenvalue probability density function proportional to 
\begin{equation}\label{A1}
\prod_{l=1}^N w(x_l) \mathbbm 1_{x_l \in I} \prod_{1 \le j < k \le N} (x_k - x_j)^2.
    \end{equation}
    For $\{p_j(x) \}$ the corresponding (monic) orthogonal polynomials,
    $\int_I w(x) p_j(x) p_k(x) \, dx = h_j \delta_{j,k}$, one has
    that the eigenvalue density is given by \cite[Eq.~(5.13)]{Fo10} 
\begin{align}\label{k0}
    \rho_{(1),N} (x)=K_N(x,x)=\frac{w(x)}{h_{N-1}}(p_{N}'(x)p_{N-1}(x)-p_{N}(x)p_{N-1}'(x)).
\end{align}
Here, the notation $K_N(x,x)$ comes from the right-hand side being derived as the limit $x \to y$ of the so-called correlation kernel $K_N(x,y)$
summed according to the Christoffel--Darboux formula; see, e.g.,~\cite[Prop.~5.1.3]{Fo10}.
Our aim is to establish that, for
$w(x)$ one of the Gaussian, Laguerre, Jacobi or Cauchy weights from the theory of classical orthogonal polynomials, the density 
$ \rho_{(1),N} (x)$ satisfies a third order linear differential equation which can be written in a structured form common to all these classical cases. Although this third order differential equation is known from earlier work --- see \cite{GT05,WF14,RF21} for the Gaussian case;  \cite{GT05,ATK11,RF21} for the Laguerre case;
\cite{RF21} for the Jacobi case; \cite{FR21} for the Cauchy case ---
none of the previous derivations apply to all the classical cases simultaneously. In the Gaussian case, our derivation reduces to that given in \cite{WF14}.

The characterising feature of the classical weight functions is that they satisfy Pearson's equation
\begin{align}\label{Pearson-eq}
    (\sigma(x)w(x))'=\tau(x)w(x),
\end{align}
where $\sigma(x)$ and $\tau(x)$ are polynomials with  $\deg \sigma(x)\leq 2$ and $\deg \tau(x)\leq 1$. Explicit forms for $\sigma(x)$ and $\tau(x) $ are listed in Table \ref{tab:pearson_params}.
The corresponding classical orthogonal polynomials $\{p_{j}(x)\}_{j\in \mathbb{N}}$ are uniquely determined, up to normalization, as the polynomial solutions to the differential equation
\begin{align}\label{diff-eq-pN}
    (\sigma(x)w(x)p_N'(x))'=\lambda_Nw(x)p_N(x).
\end{align}
Note that here $\lambda_N$, which can be viewed as an eigenvalue,
can be determined by comparing the leading coefficients of the polynomials in the equation. In the general framework, we fix the normalization by requiring $p_N(x)$ to be monic; while specializing to one of the classical cases, we adopt the common convention of the normalization.
\begin{table}[h]
\centering
\renewcommand{\arraystretch}{1.5}
\begin{tabular}{|l|c|c|c|c|}
\hline
\textbf{Polynomial} & \textbf{Interval} & \textbf{Weight} $w(x)$ & $\sigma(x)$ & $\tau(x)$ \\
\hline
Hermite $H_N(x)$ & $(-\infty, \infty)$ & $e^{-x^2}$ & $1$ & $-2x$ \\
\hline
Laguerre $L_N^{(\alpha)}(x)$ & $[0, \infty)$ & $x^\alpha e^{-x}$ & $x$ & $1+\alpha-x$ \\
\hline
Jacobi $P_N^{(\alpha, \beta)}(x)$ & $[-1, 1]$ & $(1-x)^\alpha (1+x)^\beta$ & $1-x^2$ & $\beta - \alpha - (\alpha + \beta + 2)x$ \\
\hline
R.-R. $R_N^{(\eta,\bar{\eta})}(x)$ &  $(-\infty, \infty)$ & $(1 - i x)^\eta(1+ix)^{\bar{\eta}}$ & $1+x^2$ & $2(q+(1-N-p)x)$ \\
\hline
\end{tabular}
\caption{Weight and Pearson pair for classical orthogonal polynomials. Here, R.--R.~stands for Routh--Romanovski (see, e.g.,~the review \cite{RWAK07}) and in the corresponding weight function $\eta$ is $N$ dependent,
$\eta = - (N+p-iq)$ with $p,q$ real and $p>-1$. Adopting the notation in \cite{RWAK07}, we denote $Q_{N}^{(2q,1-p)}(x):=R_N^{(\eta,\bar{\eta})}(x) $ from now on. }
\label{tab:pearson_params}
\end{table}

\subsection{The raising and lowering operators}
It is fundamental that orthogonal polynomials satisfy the three term recurrence relation
\begin{align}
    xp_N(x)=p_{N+1}(x)+d_Np_N(x)+e_N p_{N-1}(x),
\end{align}
and with classical weights, the structure relation
\begin{align}
    \sigma(x) p_N'(x)=a_Np_{N+1}(x)+b_N p_{N}(x)+c_N p_{N-1}(x);
\end{align}
with regards to the latter see, e.g.,~\cite{AFNV00}.
In these equations, $a_N,b_N,c_N,d_N$ and $e_N$ are certain constants that can be determined by comparing coefficients on both sides.
Eliminating one of the three terms on the right-hand side results in the following equations
\begin{align}
    &\bigg(\sigma(x) \frac{d}{dx}+a_N(d_N-x)-b_N\bigg)p_N(x)=(c_N-a_Ne_N)p_{N-1}(x),\label{lower}\\
    &\bigg(\sigma(x) \frac{d}{dx}+\frac{c_N}{e_N}(d_N-x)-b_N\bigg)p_N(x)=\bigg(a_N-\frac{c_N}{e_N}\bigg)p_{N+1}(x)\label{raise}.
\end{align}
The operators in the brackets on the left-hand side are usually called the lowering and raising operators, respectively. For the cases we are considering, these relations can be specified by the following:
\begin{itemize}
    \item \textbf{Hermite Polynomials} $H_n(x)$
    \[
    \begin{aligned}
        H_n'(x) &= 2n H_{n-1}(x), \\
        H_n'(x) &= 2x H_n(x) - H_{n+1}(x).
    \end{aligned}
    \]

    \item \textbf{Laguerre Polynomials} $L_n^{(\alpha)}(x)$
    \[
    \begin{aligned}
        x L_n^{(\alpha)\prime}(x) &= n L_n^{(\alpha)}(x) - (n+\alpha) L_{n-1}^{(\alpha)}(x), \\
        x L_n^{(\alpha)\prime}(x) &= (n+1) L_{n+1}^{(\alpha)}(x) - (n+\alpha+1-x) L_n^{(\alpha)}(x).
    \end{aligned}
    \]

    \item \textbf{Jacobi Polynomials} $P_n^{(\alpha, \beta)}(x)$
    \[
    \begin{aligned}
        (2n+\alpha+\beta)(1-x^2) P_n^{(\alpha, \beta)\prime}(x) &= -n \left[ (2n+\alpha+\beta)x - (\beta-\alpha) \right] P_n^{(\alpha, \beta)}(x) \\
        &\quad + 2(n+\alpha)(n+\beta) P_{n-1}^{(\alpha, \beta)}(x), \\
        (2n+\alpha+\beta)(1-x^2) P_n^{(\alpha, \beta)\prime}(x) &= (n+\alpha+\beta) \left[ (2n+\alpha+\beta)x + (\alpha-\beta) \right] P_n^{(\alpha, \beta)}(x) \\
        &\quad - 2(n+1)(n+\alpha+\beta+1) P_{n+1}^{(\alpha, \beta)}(x).
    \end{aligned}
    \]

     \item \textbf{R.-R.~polynomials} $Q_n^{(\alpha, \beta)}(x)$
    \begin{align*}
        &Q^{(\alpha,\beta)\prime}_{n}(x)=n(2\beta+n-1)Q^{(\alpha,\beta)}_{n-1}(x),\\&(1+x^2)Q^{(\alpha,\beta)\prime}_{n}(x)=Q^{(\alpha,\beta)}_{n+1}(x)-(2(\beta-n-1)x+\alpha)Q_{n}^{(\alpha,\beta)}(x).
    \end{align*}
\end{itemize}

\subsection{Differential equations satisfied by the eigenvalue densities of classical ensembles}

Let $K_N(x,x)$ be specified by (\ref{k0}). We define the operator $\theta:f(x)\mapsto \frac{d}{dx}(\sigma(x)f(x))$, which will be used extensively to simplify notation.
As a simple consequence of equation \eqref{diff-eq-pN}, we have 
\begin{align}\label{k1}
    \theta K_N(x,x)=\frac{(\lambda_N-\lambda_{N-1})w(x)}{h_{N-1}}p_{N}(x)p_{N-1}(x).
\end{align}
Further applications of $\theta$ to $K_N(x,x)$ give
\begin{align}
    &\begin{aligned}\label{k2}
        \theta^2 &K_N(x,x)= \tau(x)\theta K_{N}(x,x)+\frac{(\lambda_N-\lambda_{N-1})\sigma(x)w(x)}{h_{N-1}}  (p_{N}'(x)p_{N-1}(x)+p_{N}(x)p_{N-1}'(x))\\
        &=\tau(x)\theta K_{N}(x,x)+(\lambda_N-\lambda_{N-1})\sigma(x)K_N(x,x)\\
        &\qquad +\frac{2(\lambda_N-\lambda_{N-1})\sigma(x)w(x)}{h_{N-1}}p_{N}(x)p_{N-1}'(x),
    \end{aligned}\\
    &\begin{aligned}\label{k3}
        \theta^3K_N&=\tau'\sigma\theta K_N+\tau \theta^2 K_N\\
        &\qquad +\frac{(\lambda_N-\lambda_{N-1})w}{h_{N-1}}(\sigma'(p_Np_{N-1})'+(\lambda_N+\lambda_{N-1})\sigma p_Np_{N-1}+2\sigma^2 p_N'p_{N-1}')\\
    &=\tau'\sigma\theta K_N+\tau \theta^2 K_N+\sigma'(\theta^2 K_N-\tau \theta K_N)+(\lambda_N+\lambda_{N-1})\sigma \theta K_N\\
    &\qquad +\frac{2(\lambda_N-\lambda_{N-1})\sigma^2w}{h_{N-1}}p_{N}'p_{N-1}',
    \end{aligned}
\end{align}
where $\tau$ and $\sigma$ are the polynomials in Pearson's equation \eqref{Pearson-eq}.
Equations \eqref{k0},\eqref{k1},\eqref{k2} and \eqref{k3} are linear equations of five unknown quantities $p_N'p_{N-1}',p_Np_{N-1}',p_N'p_{N-1},p_Np_{N-1}$, $K_N$. In order to obtain a differential equation for $K_N$, we need an extra equation to form a closed system.
This can be accomplished through the raising and lowering operators for the classical orthogonal polynomials.  By multiplying together both sides of equation \eqref{lower} and equation \eqref{raise}, with the latter having $N$ replaced by $N-1$, we obtain
\begin{align}\label{eq5}
    \begin{aligned}
        (c_N&-a_Ne_N)(a_{N-1}e_{N-1}-c_{N-1}) p_{N-1}p_N\\
    =&(c_N(d_N-x)-b_Ne_N)\sigma p_Np_{N-1}'+(a_{N-1}(d_{N-1}-x)-b_{N-1})e_{N}\sigma p_N'p_{N-1}\\
    &+(c_{N}(d_N-x)-b_N e_N)(a_{N-1}(d_{N-1}-x)-b_{N-1})p_Np_{N-1}+e_N \sigma^2p_N'p_{N-1}'.
    \end{aligned}
\end{align}
This can be further specified according to the following:
\begin{itemize}
    \item \textbf{Hermite Polynomials} $H_N(x)$
    \[
    H_N' H_{N-1}' - 2x H_N' H_{N-1}= -2N H_NH_{N-1}.
    \]
    \item \textbf{Laguerre Polynomials} $L_N^{(\alpha)}(x)$
    \[
    x L_N^{(\alpha)\prime} L_{N-1}^{(\alpha)\prime} + (N+\alpha-x) L_N^{(\alpha)\prime} L_{N-1}^{(\alpha)} - N L_N^{(\alpha)} L_{N-1}^{(\alpha)\prime} = -N L_N^{(\alpha)} L_{N-1}^{(\alpha)}.
    \]

    \item \textbf{Jacobi Polynomials} $P_N^{(\alpha, \beta)}(x)$ \\
    Let $\mathcal{D}_N = 2N+\alpha+\beta$, we have
    \[
    \begin{aligned}
        \mathcal{D}_N(\mathcal{D}_N-2)(1-x^2) &P_N^{(\alpha, \beta)\prime} P_{N-1}^{(\alpha, \beta)\prime} 
         + N(\mathcal{D}_N-2) \left[ \mathcal{D}_N x + (\alpha-\beta) \right] P_N^{(\alpha, \beta)} P_{N-1}^{(\alpha, \beta)\prime} \\
         &\qquad +4N(N+\alpha)(N+\beta)(N+\alpha+\beta-1) P_N^{(\alpha, \beta)} P_{N-1}^{(\alpha, \beta)}\\
        &=\mathcal{D}_N (N+\alpha+\beta-1) \left[ (\mathcal{D}_N-2)x + (\alpha-\beta) \right] P_N^{(\alpha, \beta)\prime} P_{N-1}^{(\alpha, \beta)}.
    \end{aligned}
    \]
     \item \textbf{R.-R.~polynomials} $Q^{(\alpha,\beta)}_{N}(x)$
     \begin{align*}
         (x^2+1)Q_{N-1}^{(\alpha,\beta)\prime}Q_{N}^{(\alpha,\beta)\prime}+(2(\beta-N)x+\alpha)Q_{N-1}^{(\alpha,\beta)}Q_{N}^{(\alpha,\beta)\prime}=2(2\beta+N-1)Q_{N-1}^{(\alpha,\beta)}Q_{N}^{(\alpha,\beta)}.
     \end{align*}
\end{itemize}
Therefore, eliminating the terms from equation \eqref{k0},\eqref{k1},\eqref{k2},\eqref{k3} and \eqref{eq5} leads to a third order differential equation for the density $\rho_{(1),N}(x) $ of the classical unitary ensembles.

\section*{Appendix B: Direct determination of \texorpdfstring{$\{u_j(\gamma)\}$}{u_j(γ)}}
\renewcommand{\thesection}{B} 
\renewcommand{\theHsection}{B} 
\setcounter{equation}{0}
\setcounter{theorem}{0}

The bilateral Laplace transform of the soft edge scaled GUE density is defined by
\begin{equation}\label{B.1}
    F_N(\gamma):= \int_{-\infty}^\infty e^{\gamma x} \rho_{(1),N}^{\rm GUE,s}(x) \, dx, \quad \rho_{(1),N}^{\rm GUE,s}(x) := {1 \over \sqrt{2} N^{1/6}}\rho_{(1),N}^{\rm GUE}(\sqrt{2N} + x/(\sqrt{2} N^{1/6})).
    \end{equation}
 Substituting in (\ref{2.2}) and recalling (\ref{2.6}) we see that for large $N$,
 \begin{equation}\label{B.1a}
  F_N(\gamma) = u_0(\gamma) + {1 \over N^{2/3}} u_1(\gamma) +
  {1 \over N^{4/3}} u_2(\gamma) + \cdots.
  \end{equation} 
Here, we would like to undertake a direct computation of the large $N$ expansion of $F_N(\gamma)$, and thus of $\{u_j(\gamma)\}$ (at least for small $j$).

We begin by a simple change of variables to obtain from (\ref{B.1}) that
\begin{equation}\label{B.2}
 F_N(\gamma) = e^{-2 \gamma N^{2/3}} \int_{-\infty}^\infty e^{\gamma \sqrt{2} N^{1/6} y} \rho_{(1),N}^{\rm GUE}(y) \, dy.
  \end{equation}
  Using now a result of Ullah from 1985 \cite{Ul85} for the Fourier transform of the GUE density we can evaluate the integral to obtain
 \begin{equation}\label{B.3} 
  F_N(\gamma) = e^{-2 \gamma N^{2/3}} e^{\gamma^2 N^{1/3} / 2} 
  L_{N-1}^{(1)}(-\gamma^2 N^{1/3}),
   \end{equation}
   where $L_N^{(\alpha)}(x)$ denotes the Laguerre polynomial. The task now is to obtain the large $N$ expansion of the particular Laguerre polynomial in (\ref{B.3}).

   For this we introduce the contour integral form of the Laguerre polynomial
  \begin{equation}\label{B.4}   
  L_{N-1}^{(1)}(-z) = {e^{-z} \over 2 \pi i} {1 \over c} \oint_{|(s/c)-1|=R} \Big ( 1 - {c \over s} \Big )^{-N} e^{z s/c} \, ds,
   \end{equation}
   where $R>0, c> 0$ are arbitrary; see \cite[\S 2]{BBC08}. This with $c=\gamma/N^{1/3}$ and $z = \gamma^2 N^{1/3}$ substituted in (\ref{B.3})
   gives
   \begin{equation}\label{B.5} 
  F_N(\gamma) = e^{- 2\gamma N^{2/3}} e^{\gamma^2 N^{1/3} / 2} 
  {N^{1/3} \over \gamma} {1 \over 2 \pi i} 
  \oint_{|(N^{1/3} s/\gamma)-1|=R} \Big ( 1 - {\gamma \over N^{1/3} s} \Big )^{-N} e^{\gamma N^{2/3} s} \, ds.
   \end{equation}
We remark that from this integral representation, the large $N$ expansion (\ref{B.1a}) is far from evident. Nonetheless, we will show that at least for small orders it provides a systematic computation scheme.

With a view to applying the saddle point method, we expand
 \begin{equation}\label{B.6}
 \Big ( 1 - {\gamma \over N^{1/3} s} \Big )^{-N}  =
 \exp \bigg ( N \sum_{l=1}^\infty {1 \over l}  \Big ( {\gamma \over N^{1/3}s} \Big )^l \bigg ). 
 \end{equation}
 Substituting in (\ref{B.5}) shows that the leading order exponential factors  in the integrand are $e^{\gamma N^{2/3}(s + 1/s)} $. This has critical points $s=\pm 1$, of which $s=1$ maximises these factors. Hence, according to the saddle point method, the contour of integration should be deformed to pass through this point. Since the second derivative of $s + 1/s$ is positive, the direction of steepest descent is parallel to the imaginary axis, which tells us to introduce the variable $t$ according to
 $(s-1) = it/(N^{1/3} \gamma^{1/2})$. Doing this, the leading order exponential factors reduce to $e^{2\gamma N^{2/3}} e^{-t^2}$ times terms lower order in $N$. An asymptotic expansion to all algebraic orders in $N$ requires that the full expression (\ref{B.6}) be similarly expanded about $s=1$.

 With
  \begin{equation}\label{B.7}
  \alpha_{k,p} := \sum_{l=p+1}^\infty {(l)_k \over l} \Big ( {\gamma \over N^{1/3}} \Big )^l,
  \end{equation} 
a straightforward calculation gives from the expansion of the integrand about $s=1$, and the introduction of the variable $t$, the asymptotic expression
 \begin{multline}\label{B.8}
 F_N(\gamma) \sim {e^{\gamma^3/12} \over 2 \pi \gamma^{3/2}}
 \int_{-\infty}^\infty e^{- (t + i \gamma^{3/2}/2)^2} e^{N(\alpha_{0,3} + \alpha_{1,2} (-it/(N^{1/3}\gamma^{1/2})) + \alpha_{2,1} (-it/(N^{1/3} \gamma^{1/2}))^2/2!)} \\
 \times \prod_{j=3}^\infty e^{N \alpha_{j,0} (-it/(N^{1/3} \gamma^{1/2}))^j/j!} \, dt.
 \end{multline}
 All factors in the integrand, apart from the first, go to zero like powers of $N^{-1/3}$. Explicitly, to first order in $N^{-1/3}$, we calculate
   \begin{equation}\label{B.9}
 F_N(\gamma) \sim {e^{\gamma^3/12} \over 2 \pi \gamma^{3/2}}
 \int_{-\infty}^\infty e^{- (t + i \gamma^{3/2}/2)^2} \Big ( 1 +
 N^{-1/3} \Big ( {\gamma^4 \over 4} - i \gamma^{5/2} t - {3 \over 2} \gamma t^2 +
 {i \over \gamma^{1/2} } t^3 \Big ) \Big )\, dt.
  \end{equation} 
  Changing variables $t \mapsto t - i  \gamma^{3/2}/2$, and noting that all odd monomials integrate to zero against $e^{-t^2}$, we are left with only the even monomials to consider. For the term of order $N^{-1/3}$, as can be read off from (\ref{B.9}), 
  these terms are seen to add up to zero.
  Thus, we conclude that there is no term in the asymptotic expansion at order $N^{-1/3}$, in agreement with (\ref{2.2}). 

  The calculation of the asymptotic expansion to higher orders in $N^{-1/3}$ starting from (\ref{B.8}) is well suited to the use of computer algebra. In particular, expanding up to and including order $N^{-2}$, we reclaim $u_1(\gamma)$, $u_2(\gamma)$ as presented in (\ref{C1.1}), and further $u_3(\gamma)$ as specified by (\ref{2.7}) with $b$ given by (\ref{2.8}).

  \begin{remark}
  For large $N$ but
 fixed $z$ it is known \cite{VA01,BBC08} that 
   \begin{equation}\label{B.10}
   L_N^{(-a)}(-z) \sim {e^{-z/2} \over 2 \sqrt{\pi}} 
   {e^{2 \sqrt{Nz}} \over z^{1/4 - a/2} N^{1/4+a/2}} \bigg (
   1 + \Big ( {3 -12 a^2 + 24(1-a) z + 4 z^2 \over 48 \sqrt{z}}
   \Big ) {1 \over N^{1/2}} + {\rm O} \Big ( {1 \over N}\Big ) \bigg ).
   \end{equation}
   In comparison, from (\ref{B.3}) and (\ref{B.9}) we have
   \begin{equation}\label{B.11}
   L_N^{(1)}(-\gamma^2 N^{1/3}) \sim  e^{2 \gamma N^{2/3}}
   e^{-\gamma^2 N^{1/3}/2} {e^{\gamma^3/12} \over 2 \sqrt{\pi} \gamma^{3/2}}
   \Big ( 1 + {\rm O}(N^{-2/3}) \Big ).
   \end{equation} 
   Formally substituting $z = \gamma^2 N^{1/3}$ in (\ref{B.10}) with $a=-1$
   gives
     \begin{equation}\label{B.12}
      L_N^{(1)}(-\gamma^2 N^{1/3}) \sim  e^{2 \gamma N^{2/3}}
   e^{-\gamma^2 N^{1/3}/2} {1 \over 2 \sqrt{\pi} \gamma^{3/2}}
   \Big ( 1 + {\gamma^3 \over 12} + \cdots \Big ).
    \end{equation}
    This is consistent with (\ref{B.11}) if it should be that each successive term in the $N^{-1/2}$ expansion inside the main brackets in
    (\ref{B.10}) should give successive terms in the power series expansion of $e^{\gamma^3/12}$, as is the case with the displayed term.
     
  \end{remark}

\section*{Appendix C: A \texorpdfstring{$\sigma$-Painlev\'e III$'$}{σ-Painlev´e III′} transcendent verification of (\ref{1.11k})}

\renewcommand{\thesection}{C} 
\renewcommand{\theHsection}{C} 
\setcounter{equation}{0}
\setcounter{theorem}{0}

Let $\sigma_0(x;\xi)$ be the $\sigma$-Painlev\'e III$'$ transcendent specified as the solution of the particular 
$\sigma$-Painlev\'e III$'$ equation
\begin{equation}\label{C1}
(x \sigma'')^2 + \sigma'(1 + 4 \sigma')(x \sigma' - \sigma) = (a \sigma')^2
\end{equation}
subject to the boundary condition
\begin{equation}\label{C3}
\sigma_0(x;\xi) \mathop{\sim}\limits_{x \to 0^+} - \xi x r_0^{\rm L,h}(x).
\end{equation}
Let $\sigma_1(x;\xi)$ be characterised as satisfying a particular second order linear differential equation,
\begin{equation}\label{2.17D}
A(x) \sigma_1'' + B(x) \sigma_1' + C(x) \sigma_1 = D(x),
\end{equation}
subject to the boundary condition
\begin{equation}\label{C4}
\sigma_1(x;\xi) \mathop{\sim}\limits_{x \to 0^+} - \xi x r_1^{\rm L,h}(x).
\end{equation}
In (\ref{2.17D}), the functions $A(x),\dots,D(x)$ depend on $\sigma_0$ and/or its first and second derivative; their explicit forms are given in \cite[Eqns.~(2.27) and (2.31)]{FT19}.

We have from \cite{TW94a,FW02} that
\begin{equation}\label{C5}
\mathcal E_{2}^{\rm h}((0,y);\xi)
= \exp\Big ( \int_0^y \sigma_0(x;\xi)
\, {dx \over x} \Big ),
\end{equation}
and from \cite{FT19} that
\begin{equation}
\mathcal E_{2}^{1,\rm L, h}((0,y);\xi)
= \exp\Big ( \int_0^y \sigma_1(x;\xi)
\, {dx \over x} \Big )
\end{equation}
(in \cite{FT19} $\sigma_1(x;\xi)$
is denoted
$\hat{\sigma}_2(x;\xi)$). 
Substituting in (\ref{1.11k}) shows that the latter is valid provided
\begin{equation}\label{C8}
\bigg ( {1 \over 48} y + {a^2 - 2 \over 24} \bigg ) \sigma_0 + {1 \over 12} \Big ( y \sigma_0' + \sigma_0^2 \Big ) = \int_0^y {\sigma_1(x;\xi) \, {dx \over x} }.
\end{equation}
Differentiating this expresses $\sigma_1(y;\xi)$ in terms of $\{y,\sigma_0,\sigma_0',\sigma_0''\}$. Substituting (\ref{C3}) and (\ref{C4}) regarded as functional forms for small $\xi$ gives precisely the differential relation 
(\ref{4.22a}) expressing $r_1^{\rm L,h}(y)$ in terms of $r_0^{\rm L,h}(y)$. It then remains to establish (\ref{2.17D}) with $\sigma_1(y;\xi)$ so expressed. 
Analogous tasks have been met in \cite{FPTW19,FM23,FS25}. The new form of (\ref{2.17D}) is an identity involving $y, \sigma_0$ and its first four derivatives. The third and fourth derivatives can be eliminated by differentiating (\ref{C8}) an appropriate number of times. This leaves a
differential identity involving $\{y,\sigma_0,\sigma_0',\sigma_0''\}$. The use of computer algebra allows for the latter to be factorised, with one of the factors vanishing due to (\ref{C1}).

\section*{Appendix D: Hard edge scaling and hypergeometric polynomials}

\renewcommand{\thesection}{D} 
\renewcommand{\theHsection}{D} 
\setcounter{subsection}{0}
\setcounter{equation}{0}
\setcounter{theorem}{0}

\subsection{Scaled large \texorpdfstring{$N$}{N} expansion of \texorpdfstring{$L_N^{(\alpha)}(x/N)$}{LN(α)(x/N)}}

The Laguerre polynomials are expressed as particular hypergeometric polynomials according to
\begin{equation}\label{E1}
L_N^{(\alpha)}(x) = \binom{N+\alpha}{N}
\, {}_1 F_1 (-N;\alpha+1;x).
\end{equation}
Here, we consider the large $N$ expansion of the hypergeometric polynomial in (\ref{E1}) with $x \mapsto x/N$, up to and including terms of order $1/N^4$.

\begin{proposition}\label{PE1}
Let ${\sf D}_x^p := x^p {d^p \over d x^p}$. For large $N$, we have
\begin{multline}\label{E2}
 {}_1 F_1 (-N;\alpha+1;x/N) = \bigg ( 1 - {1 \over N}  {1 \over 2} {\sf D}_x^2 
 +
 {1 \over N^2} \Big ( {1 \over 8} {\sf D}_x^4 +  {1 \over 3} {\sf D}_x^3 \Big ) 
 -  {1 \over N^3} \Big ( {1 \over 4 8} {\sf D}_x^6 +  {1 \over 6} {\sf D}_x^5  +  {1 \over 4} {\sf D}_x^4 \Big ) \\
 + {1 \over N^4} \Big ( {1 \over 384} {\sf D}_x^8 +  {1 \over 24} {\sf D}_x^7  +  {13 \over 72} {\sf D}_x^6  + {1 \over 5}  {\sf D}_x^5\Big ) + \cdots \bigg ) \,  {}_0 F_1 (\alpha+1;-x)
 \end{multline}
 and
 \begin{multline}\label{E2a}
 {}_1 F_1 (-N+1;\alpha+1;x/N) = \bigg ( 1 - {1 \over N} \Big ( {1 \over 2} {\sf D}_x^2 + {\sf D}_x \Big ) 
 + {1 \over N^2 } \Big ( {1 \over 8} {\sf D}_x^4 +  {5 \over 6} {\sf D}_x^3  + {\sf D}_x^2 \Big ) 
 -  {1 \over N^3} \Big ( {1 \over 4 8} {\sf D}_x^6 +  {7 \over 24} {\sf D}_x^5   \\   +  {13 \over 12} {\sf D}_x^4  + {\sf D}_x^3 \Big )
 + {1 \over N^4} \Big ( {1 \over 384} {\sf D}_x^8 +  {1 \over 16} {\sf D}_x^7  +  {17 \over 36} {\sf D}_x^6  + {77 \over 60}  {\sf D}_x^5  + {\sf D}_x^4 \Big ) + \cdots \bigg ) \,  {}_0 F_1 (\alpha+1;-x)
 \end{multline}
\end{proposition}

\begin{proof}
Consider first (\ref{E2}). We use the facts that
$$
{}_1 F_1 (-N;\alpha+1;x/N) = \sum_{p=0}^\infty {(-N)_p \over p! (\alpha+1)_p} \Big ( {x \over N} \Big )^p, \quad 
{(-N)_p \over N^p} = (-1)^p \prod_{l=1}^{p-1} \Big (
1 - {l \over N} \Big ),
$$
and then the expansion
\begin{equation}\label{E3}
    \prod_{l=1}^{p-1} \Big (
1 - {l \over N} \Big ) = 1 - {1 \over N} \sum_{l=1}^{p-1} l +
{1 \over N^2} \sum_{1 \le l_1 < l_2 \le p-1} l_1 l_2 -{1 \over N^3}
 \sum_{1 \le l_1 < l_2 < l_3 \le p-1} l_1 l_2 l_3 + \cdots .
\end{equation}
The sums, with the assistance of computer algebra, can be expressed in closed form, giving a polynomial in $p$ of degree $2k$ for the term relating to $N^{-2k}$. Substituted inside the summand, the resulting expressions are recognised in terms of derivatives of the limiting expression $ {}_0 F_1 (\alpha+1;-x)$ as specified.

The derivation of (\ref{E2a}) proceeds analogously, now making use of (\ref{E3}) with the upper terminal $p-1$ replaced by $p$.
\end{proof}

\subsection{Functional property of \texorpdfstring{$K_N^{\rm LUE}(x/4N_{{\rm s},a}',y/4N_{{\rm s},a}')$}{the hard edge scaled kernel}}
From the Christoffel--Darboux formula \cite[Eq.~(5.10)]{Fo10} with the quantities therein given as specified in \cite[\S 5.4.1]{Fo10}, then with the Laguerre polynomials substituted by  use of (\ref{E1}), the LUE correlation kernel permits the explicit form
\begin{multline}\label{E3a}
K_N^{\rm LUE}(x,y) = 
 {\Gamma(N+a+1) \over (\Gamma(a+1))^2 \Gamma(N)} {(xy)^{a/2} e^{-(x+y)/2} \over x - y} \Big (
  {}_1 F_1 (-N;a+1;y)  {}_1 F_1 (-N+1;a+1;x)  \\
 - {}_1 F_1 (-N;a+1;x)  {}_1 F_1 (-N+1;a+1;y)
    \Big ).
  \end{multline}
  Use of the Kummer transformation
  \begin{equation}\label{E3b}
   {}_1 F_1 (a;c;x) = e^{x} {}_1 F_1 (c-a;c;-x)
 \end{equation}
 allows for a functional equation to be established relating to $K_N^{\rm LUE}(x,y)$ with hard edge scaling.

 \begin{proposition}\label{PE2}
     Let $N_{\rm h}'=N + a/2$, and specify $K_N^{\rm LUE}(x,y)$ by the right-hand side of (\ref{E3a}). We have
   \begin{equation}\label{E3c}  
  {1 \over 4N_{\rm h}'} K_N^{\rm LUE}\Big ( {x \over 4N_{\rm h}'},{y \over 4N_{\rm h}'} \Big ) =   {1 \over 4N_{\rm h}'} K_N^{\rm LUE} \Big ( {x \over 4N_{\rm h}'},{y \over 4N_{\rm h}'} \Big ) \Big |_{N_{\rm h}' \mapsto - N_{\rm h}'}.
   \end{equation}
   Hence, the $1/N_{\rm h}'$ expansion must proceed via powers of 
   $1/(N_{\rm h}')^2$.
 \end{proposition}

 \begin{proof}
     We read off from (\ref{E3c}) that the $N,N_{\rm h}'$ dependent prefactors coming from the factor of $(4N_{\rm h}')^{-1}$ in (\ref{E3c}) and the factor of
     \begin{equation*}
        \frac{\Gamma(N+a+1)(xy)^{a/2}}{(\Gamma(a+1))^2\Gamma(N)}
     \end{equation*}
     in (\ref{E3a}) are
     $$
     \Big ( {1 \over 4 N_{\rm h}'} \Big )^{a+1} {\Gamma(N+a+1) \over  \Gamma(N)} =  \Big ( {1 \over 4 N_{\rm h}'} \Big )^{a+1}  {\Gamma(N_{\rm h}'+a/2+1) \over  \Gamma(N_{\rm h}'-a/2)},
     $$
     where the equality follows from the relation $N_{\rm h}' = N +a/2$. Temporarily assuming that $a$ is an integer shows that the ratio of gamma functions reduces to $\prod_{k=0}^a(N_{\rm h}'-a/2+k)$, which when divided by $(4N_{\rm h}')^{a+1}$ can be verified to return a quantity even in $N_{\rm h}'$; for the removal of the assumption on $a$, see the text relating to (\ref{EZ}) above. That the remaining $x,y$ dependent terms in (\ref{E3a}) as implied by the left-hand side of (\ref{E3c}) are even in $N_{\rm h}'$ is immediate upon the definition of $N_{\rm h}'$ and use of the Kummer transformation (\ref{E3b}).
 \end{proof}

\section*{Appendix E: Details of the proof of Proposition \ref{PEv}}

\renewcommand{\thesection}{E} 
\renewcommand{\theHsection}{E} 
\setcounter{equation}{0}
\setcounter{theorem}{0}

 We focus attention on (\ref{2.RLy}).
 We know from (\ref{4.16}) that
 \begin{equation}
 {1 \over 4 (N - 1 + a/2)}   \rho_{(1),N-1}^{\rm L U E,h}\Big ( {y \over 4 (N-1+a/2)} \Big )=
 r_0^{\rm L,h}(y) + {1 \over (N-1+a/2)^2} r_1^{\rm L,h}(y) + \cdots.
     \end{equation}
     Recalling that for $\beta = 1$, $\hat{N}_{\rm h}' = N + (a-1)/2$, it follows that
\begin{equation}\label{5.29b}
 {1 \over 4 \hat{N}_{\rm h}'} \rho_{(1),N-1}^{\rm L U E,h}\Big ( {y \over 4 \hat{N}_{\rm h}'} \Big ) =
     \Big ( 1 - {1 \over 2 \hat{N}_{\rm h}'} \Big ) 
 r_0^{\rm L,h} \bigg (  \Big ( 1 - {1 \over 2 \hat{N}_{\rm h}'} \Big ) y \bigg ) +    {1 \over ( \hat{N}_{\rm h}')^2} 
  r_1^{\rm L,h}(y)  + \cdots.
  \end{equation}
 With knowledge of the explicit functional forms of $r_i^{\rm h}(y)$ ($i=0,1)$ from (\ref{4.17}), the sought expansion can be carried out with the aid of computer algebra. As found in the context of Proposition \ref{P4.7}, final paragraph of the proof, the Mathematica output for this step involves $J_{a-1}(\sqrt{y})$, which we substitute for using the identity $J_{a-1}(u) = (a/u) J_a(u) + J_a'(u)$.

 Using (\ref{E1}), we have
 $$
 {1 \over 4 \hat{N}_{\rm h}'} {(N-1)! \over  \Gamma(a-1+N)}
  L_{N-1}^{(a)}
  \Big ({y \over 4  \hat{N}_{\rm h}'} \Big ) =
  {N+a-1 \over 4 \Gamma(a+1) \hat{N}_{\rm h}'} \,
  \, {}_1 F_1 \Big ( - N +1 ; a+1; -{y \over 4\hat{N}_{\rm h}'} \Big ).
 $$
 Making use of (\ref{E2}), we see that this
 multiplied by $y^{(a-1)/2}e^{-y/4\hat{N}_{\rm h}'}$
 expands for large
 $\hat{N}_{\rm h}'$ as 
  \begin{multline}\label{5.30}
  {1 \over 4 \Gamma(a+1)} \bigg ( 1 +  {a-1 \over  2 \hat{N}_{\rm h}'}  \bigg ) \Big ( 1 -
  {y \over 2 (4 \hat{N}_{\rm h}')} + {y^2 \over 8  (4 \hat{N}_{\rm h}')^2} + \cdots \Big )\\
  \times \bigg ( 
   1 - {1 \over (N-1)}  {1 \over 2} {\sf D}_x^2 
 + {1 \over (N-1)^2 } \Big ( {1 \over 8} {\sf D}_x^4 +  {1 \over 3} {\sf D}_x^3 \Big ) + \cdots \bigg )
   {}_0 F_1 (   a+1; -x) \Big |_{x = y(N-1)/(4\hat{N}_{\rm h}')},
  \end{multline}
 although the expansion of the second line in inverse powers of $\hat{N}_{\rm h}'$ has not been finalised. The required working is best left for a computer algebra computation. Doing so, while also using the identity $J_{a+1}(u) = (a/u) J_a(u) - J_a'(u)$ in the final output, gives for the expression in total,
  \begin{multline}\label{5.30a}
   2^{a-2} y^{-1/2} \bigg ( J_a(\sqrt{y}) + {(a-1) \over 2 \hat{N}_{\rm h}'} J_a(\sqrt{y})  + {1 \over 96} {1 \over  (\hat{N}_{\rm h}')^2} \\
\times
  \Big ( (2 a (a^2-1) + y )  J_a(\sqrt{y})
  - (2(a^2 - 1) + y)  \sqrt{y} J_{a}'(\sqrt{y}) \Big ) + \cdots \bigg ).
   \end{multline}

 Similarly, in relation to the integral in the final line of
 (\ref{2.RLy}) (including the factor of 2, and after multiplication by $(4 \hat{N}_{\rm h}')^{-(a-1)/2}$
 which comes from $x^{(a-1)/2} \mapsto (x/4 \hat{N}_{\rm h}')^{(a-1)/2}$ from the line above)
 we have that it expands as
  \begin{multline}\label{5.32}
  2 ( 4\hat{N}_{\rm h}')^{-a}
  \binom{N-2+a}{N-2} \int_0^{y} x^{(a-1)/2} \Big ( 1 -
  {x \over 2 (4 \hat{N}_{\rm h}')} + {x^2 \over 8  (4 \hat{N}_{\rm h}')^2} + \cdots \Big ) \\
  \times \bigg ( 
 1 - {1 \over (N-2)}  {1 \over 2} {\sf D}_u^2 
 + {1 \over (N-2)^2 } \Big ( {1 \over 8} {\sf D}_u^4 +  {1 \over 3} {\sf D}_u^3 \Big )  + \cdots \bigg )
   {}_0 F_1 (   a+1; -u) \Big |_{u = x^*} \, dx,
   \end{multline} 
   where $x^*:= x(N-2)/(4\hat{N}_{\rm h}')$, although
   the expansion inside the integral is still to be finalised. This is analogous to the task faced in
 (\ref{5.30}), which we know can be carried out using computer algebra. Doing this gives for the integral
  \begin{multline}\label{5.33b}
2^{a+1} \Gamma(a+1) \int_0^{\sqrt{y}} \bigg ( 
J_a(u) + {1 \over 2 \hat{N}_{\rm h}'}\Big (a J_a(u)
-u J_{a}'(u) \Big )+ {1 \over 96 (\hat{N}_{\rm h}')^2}  \\
\times \Big ( \Big (  - 11 u^2
 + a (22+2 a(12+a))  \Big )
 J_{a}(u)
 -(22+2 a(12+a) + u^2) u  J_{a}'(u)  \Big )
 + \cdots  \bigg ) \, du,
\end{multline} 
where as in (\ref{5.30a}) we have substituted for $J_{a+1}(y)$ as appeared in the output so as to involve
$J'_a(\sqrt{y})$, and we have further changed variables $u = \sqrt{x}$.

In fact, the integral in (\ref{5.33b}) can be evaluated, up to an additive term proportional to ${\rm JI}_a(\sqrt{y})$
at each order.
The first step is to apply integration by parts to the terms involving $u J_{a}'(u)$. The remaining integrals are then seen to be either proportional to 
 \begin{equation}\label{5.33x}
\widetilde{{\rm JI}_a}(\sqrt{y}):=\int_0^{\sqrt{y}} J_a(u) \, du = 1- {\rm JI}_a(\sqrt{y}),
\end{equation}
or proportional to $\int_0^{\sqrt{y}} u^2 J_a(u) \, du$.
Regarding the latter, we substitute for the integrand
$(a^2 - {\sf D}_2 - {\sf D}_1)  J_a(u)$ as is permitted by the differential equation for the Bessel function (\ref{4.22}). Integration by parts then shows
  \begin{equation}\label{5.33c}
\int_0^{\sqrt{y}} u^2 J_a(u) \, du =
(a^2 - 1) \widetilde{{\rm JI}_a}(\sqrt{y}) - y J_a'(\sqrt{y})
+ \sqrt{y} J_a(\sqrt{y}).
  \end{equation}
  Thus, we compute that (\ref{5.33b}) is equal to
   \begin{multline}\label{5.33d}
2^{a+1} \Gamma(a+1) \bigg (  \widetilde{{\rm JI}_a}(\sqrt{y}) +
{1 \over 2 \hat{N}_{\rm h}'}\Big ( (a+1)\widetilde{{\rm JI}_a}(\sqrt{y})
 - \sqrt{y} 
 J_a(\sqrt{y}) \Big )
+ {1 \over 96 (\hat{N}_{\rm h}')^2}  \\
\times \Big ( 2(a+5)(a+3)(a+1) \widetilde{{\rm JI}_a}(\sqrt{y})
 -(30+2 a(12+a) + y) \sqrt{y}  J_{a}(\sqrt{y}) + 8 y J_a'(\sqrt{y}) 
  \Big )
+ \cdots   \bigg ).
\end{multline}

 For the prefactor of the integral in (\ref{5.32}), we have the expansion
  \begin{equation}\label{5.32x}
 2 ( 4 \hat{N}_{\rm h}')^{-a}
  \binom{N-2+a}{N-2} = {2^{-2a+1} \over \Gamma(a+1)}
  \Big ( 1 - {a \over \hat{N}_{\rm h}'}  - {a(a-1)(a-11) \over 24 (\hat{N}_{\rm h}')^2}
   + \cdots \Big ),
  \end{equation} 
  obtained by a calculation similar to that used to deduce
  (\ref{EZ1}). Multiplying this with (\ref{5.33d}) gives
   \begin{multline}\label{5.33y}
   2^{-a+2}  \bigg (  \widetilde{{\rm JI}_a}(\sqrt{y}) + {1 \over 2 \hat{N}_{\rm h}'}\Big ( (1-a)\widetilde{{\rm JI}_a}(\sqrt{y})
 - \sqrt{y} 
 J_a(\sqrt{y}) \Big ) + {1 \over 96 (\hat{N}_{\rm h}')^2}  \\
\times \bigg ( -2 (a-1)(a-3)(a-5)  \widetilde{{\rm JI}_a}(\sqrt{y})
 -(30+2 a(-12+a) + y) \sqrt{y}  J_{a}(\sqrt{y}) + 8 y J_a'(\sqrt{y}) 
  \bigg )
+ \cdots   \bigg ).
  \end{multline}

Remaining to be considered is the large $\hat{N}_{\rm h}'$
  expansion of the final term in the second line of (\ref{2.RLy}). After multiplication by
  $(4 \hat{N}_{\rm h}')^{-(a-1)/2} $ (recall the text above (\ref{5.32})), making use of (\ref{EZ})
  and (\ref{E3}), 
  and aided by computer algebra, we compute that this is given by
   \begin{multline}\label{5.33}
  (\hat{N}_{\rm h}')^{-(a-1)/2} 2^{-3a/2+5/2} 
   {\Gamma((N+1)/2) \Gamma(a-1+N) \over \Gamma(N)
 \Gamma((N+a)/2)} \\ = 2^{-a+2} \bigg ( 1 + {1 - a \over 2 \hat{N}_{\rm h}'} - {(a-1)(a-3)(a-5) \over 48 (\hat{N}_{\rm h}')^2} + \cdots  \bigg ) 
 \end{multline} 
 In particular, recalling (\ref{5.33x}), one sees that subtracting (\ref{5.33}) from (\ref{5.33y}) has the sole effect of changing each $\widetilde{{\rm JI}_a}(\sqrt{y})$ in the former to
 $-{{\rm JI}_a}(\sqrt{y})$.

 The final step is to multiply the said modification of (\ref{5.33y}) together with (\ref{5.30a}), multiply by $-1/4$ (this is a prefactor of the second term on the first line of (\ref{2.RLy}) and has yet to be taken into account, and then to add to (\ref{5.29b}).   With this done, we can read off the stated result.


\bigskip

\begin{thebibliography}{9}

\small

\bibitem{ATK11}
S. Adachi, M. Toda, and H. Kubotani, \emph{Asymptotic analysis of singular values of rectangular complex matrices in
the Laguerre and fixed trace ensembles}, J. Phys. A \textbf{44} (2011), 292002(8pp). 

\bibitem{AFNV00}
M. Adler, P. Forrester, T. Nagao and P. van Moerbeke, \emph{Classical skew orthogonal polynomials and random matrices}, J. Stat. Phys., \textbf{99} (2000),
141–170.

 

   \bibitem{BJ13}
 J.~Baik and R. Jenkins,  \emph{Limiting distribution of maximal crossing and nesting of Poissonized
random matchings}, Ann. Probab.,  \textbf{41} (2013), 4359--4406.


 \bibitem{BF97a}
T.H. Baker and P.J. Forrester, \emph{The {Calogero-Sutherland} model and
  generalized classical polynomials}, Commun. Math. Phys. \textbf{188} (1997),
  175--216.



\bibitem{BFP98}
T.H. Baker, P.J. Forrester, and P.A. Pearce, Random matrix ensembles with an extensive external charge, J. Phys.
A \textbf{31} (1998), 6087.


\bibitem{Bo16}
F. Bornemann, \emph{A note on the expansion of the smallest eigenvalue distribution of the LUE at the hard edge}, The
Annals of Applied Probability, \textbf{26} (2016), 1942-1946

   \bibitem{Bo24x}
 F.~Bornemann, \emph{Asymptotic expansions relating to the  lengths of longest monotone subsequences of involutions}, Experimental Math. (2024), 
 1--45. https://doi.org/10.1080/10586458.2024.2397334.
 
   \bibitem{Bo24y}
 F.~Bornemann, \emph{Asymptotic expansions relating to the distribution of the length of longest increasing subsequences}, Forum Math. Sigma. 12:e36. (2024) doi:10.1017/fms.2024.13. 


\bibitem{Bo24}
F. Bornemann,  \emph{Asymptotic expansions of the limit laws of Gaussian and Laguerre (Wishart) ensembles
at the soft edge}, arXiv:2403.07628.

\bibitem{Bo25}
F. Bornemann, \emph{Asymptotic expansions of Gaussian and Laguerre ensembles at the soft edge II: level
densities}, Random Matrices: Theory and Applications, \textbf{26} (2026), 2550025.

\bibitem{Bo25a}
F. Bornemann, \emph{Asymptotic expansions of Gaussian and Laguerre ensembles at the soft edge III:
Generating functions}, arXiv:2506.18673.

\bibitem{BFM17}
F.~Bornemann, P.J.~Forrester and A.~Mays, \emph{Finite size effects for spacing distributions in random matrix
theory: circular ensembles and Riemann zeros}, Stud. Appl. Math. \textbf{138} (2017), 401–437.

   \bibitem{BF03}
A.~Borodin and P.J. Forrester, \emph{Increasing subsequences and the
  hard-to-soft transition in matrix ensembles}, J.Phys. A \textbf{36} (2003),
  2963--2981.

\bibitem{BBC08}
D. Borwein, J.M. Borwein and R.E. Crandall,
\emph{Effective Laguerre asymptotics},
SIAM J. Numer. Anal., \textbf{46} (2008), 285-3312.




\bibitem{BEY14}
P. Bourgade, L. Erd\H{o}s, and H.-T. Yau, \emph{Edge universality of beta ensembles}, Comm. Math. Phys., \textbf{332} (2014), 261–353.

\bibitem{BF25}
S.-S. Byun and P.J. Forrester, \emph{Progress on the study of the Ginibre ensembles}, KIAS Springer Series
in Mathematics \textbf{3}, Springer, 2025.

\bibitem{BL24}
S.S.~Byun and Y.-W.~Lee, \emph{Finite size corrections for real eigenvalues of the elliptic Ginibre matrices}, Random
Matrices Theory Appl. \textbf{13} (2024), 2450005.

\bibitem{DLMS19}
D.S. Dean, P. Le Doussal, S.N. Majumdar and G. Schehr, Non-interacting fermions in a trap and random matrix theory, J. Phys. A 52, (2019),
144006.



\bibitem{DG07}
P. Deift and D. Gioev, \emph{Universality at the edge of the spectrum for unitary, orthogonal, and
symplectic ensembles of random matrices},  \textbf{60} (2007) 867–910.



\bibitem{DF06}
P.~Desrosiers and P.J. Forrester, \emph{Hermite and {L}aguerre
  $\beta$-ensembles: asymptotic corrections to the eigenvalue density}, Nucl.
  Phys. B \textbf{743} (2006), 307--332.

\bibitem{DE06}
I. Dumitriu and A. Edelman, \emph{Global spectrum fluctuations for the $\beta$-Hermite and $\beta$-Laguerre ensembles via matrix
models}, J. Math. Phys. \textbf{47} (2006), 063302. 


\bibitem{EGP16}
A. Edelman, A. Guionnet and S. P\'ech\'e, \emph{Beyond universality in random matrix theory}, The Annals of Applied
Probability, \textbf{26} (2016), 1659-1697.


\bibitem{Fo93}
P.J. Forrester,
 \emph{The spectrum edge of random matrix ensembles}, Nucl. Phys. B \textbf{402} (1993), 709–728.

 \bibitem{Fo10}
P.J. Forrester, \emph{Log-gases and random matrices}, Princeton University Press,
  Princeton, NJ, 2010.

\bibitem{forrester21}
P.J. Forrester,  \emph{Moments of the ground state density for the d-dimensional Fermi gas in an harmonic trap}, Random Matrices: Theory and Applications, \textbf{10} (2021), 2150018.

\bibitem{Fo25}
P.J. Forrester,  \emph{Dualities in random matrix theory},
arXiv:2501.07144.


\bibitem{FM23}
P.J. Forrester and A. Mays, \emph{Finite size corrections relating to distributions of the length of longest
increasing subsequences}, Adv. Applied Math. \textbf{145} (2023), 102482.


\bibitem{FNH99}
P.J.Forrester, T.Nagao,and G.Honner, \emph{Correlations for the orthogonal-unitary and symplectic-unitary transitions at the hard and soft edges}, Nucl. Phys. B \textbf{553} (1999), 601–643.


\bibitem{FPTW19}
 P.J. Forrester, J.H.H. Perk, A.K. Trinh and N.S. Witte, \emph{Leading corrections to the scaling function on
the diagonal for the two-dimensional Ising model}, J. Stat. Mech. \textbf{2019} (2019), 023106


\bibitem{FR21}
P.J. Forrester and A. Rahman, \emph{Relations between moments for the Jacobi and Cauchy random
matrix ensembles}, J. Math. Phys. \textbf{62} (2021), 073302.

\bibitem{FRS25}
P.J. Forrester, A.A.~Rahman and B.-J.~Shen, \emph{A note on optimal soft edge expansions for the Gaussian $\beta$ ensembles}, arXiv:2510.13092.

\bibitem{FRW17}
P.J. Forrester, A.A. Rahman, and N.S. Witte, \emph{Large $N$ expansions for the Laguerre and Jacobi $\beta$ ensembles from
the loop equations}, J. Math. Phys. \textbf{58} (2017), 113303.

\bibitem{FS25}
P.J.~Forrester and B.-J. Shen, \emph{Finite size corrections in the bulk for circular $\beta$ ensembles}, arXiv:2505.09865.


\bibitem{FT18}
P.J. Forrester and A.K. Trinh, \emph{Functional form for the leading correction to the distribution of the largest
eigenvalue in the GUE and LUE}, J. Math. Phys. \textbf{59} (2018), 053302.

\bibitem{FT19}
 P.J. Forrester and A. K. Trinh, \emph{Finite size corrections at the hard edge for the Laguerre $\beta$ ensemble},
Stud. Appl. Math. \textbf{143} (2019), 315–336.

\bibitem{FT19a}
 P.J. Forrester and A. K. Trinh, \emph{Optimal soft edge scaling variables for the Gaussian and
Laguerre even $\beta$ ensembles}, Nucl. Phys. B \textbf{938} (2019), 621--639.

\bibitem{FW26}
 P.J. Forrester and F.~Wei,
\emph{Higher-order linear differential equations for unitary matrix integrals: applications and generalisations (with an Appendix by Folkmar Bornemann)}, SIGMA \textbf{22}  (2026), 015, 21 pages. 

\bibitem{FW02}
P.J. Forrester and N.S. Witte, \emph{Application of the $\tau$-function theory of Painlev\'e equations to random matrices: PV,
PIII, the LUE, JUE and CUE}, Commun. Pure Appl. Math. \textbf{55} (2002), 679–727


\bibitem{GFF05}
T.M. Garoni, P.J. Forrester, and N.E. Frankel, \emph{Asymptotic corrections to the eigenvalue density of
the GUE and LUE}, J. Math. Phys. \textbf{46} (2005), 103301.


\bibitem{GT05}
F.~G\"otze and A.~Tikhomirov, \emph{The rate of convergence for spectra of
  {GUE} and {LUE} matrix ensembles}, Cent. Eur. J. Math. \textbf{3} (2005),
  666--704.

\bibitem{HT12}
  U. Haagerup and S. Thorbjørnsen. \emph{Asymptotic expansions for the Gaussian unitary ensemble}. Infin. Dimens. Anal. Quantum Probab. Relat. Top., 15(1):1250003, 41, 2012.

\bibitem{HHN16}
  W. Hachem, A. Hardy and J. Najim, \emph{Large complex correlated Wishart matrices: the Pearcey kernel and expansion
at the hard edge}, Elec. J. Probab. \textbf{21} (2016), 1–36.

\bibitem{LM79}
G.P. Lawes and N.H. March, \emph{Exact local density method for linear harmonic oscillator}, J. Chem. Phys. 71 (1979),
1007–1009.

\bibitem{Jo01}
I.M.~Johnstone, \emph{On the distribution of the largest eigenvalue in principal components
analysis}, Annals of Statistics \textbf{29} (2001), 295--327.


\bibitem{KRV16}
  M.~Krishnapur, B.~Rider, and B.~Vir\'ag, \emph{Universality of the stochastic
Airy operator}, Comm. Pure Appl. Math., \textbf{69}, (2016), 145–199.

\bibitem{Ok02}
A. Okounkov, \emph{Generating functions for intersection numbers on moduli spaces of curves}, Int. Math.
Res. Not. \textbf{18} 933–957 (2002)

\bibitem{PS16}
A. Perret and G. Schehr, \emph{Finite $N$ corrections to the limiting distribution of the smallest eigenvalue of Wishart
complex matrices}, Random Matrices: Theory and Applications, \textbf{5} (2016), 1650001


   \bibitem{RF21}
   A.A.~Rahman and P.J.~Forrester,  \emph{Linear differential equations for the resolvents of the classical matrix ensembles},
Random Matrices Theory Appl. \textbf{10} (2021), 2250003.

\bibitem{Ra98}
E.M.~Rains, \emph{Increasing subsequences and the classical groups}, Elect. J. of Combinatorics 5 (1998), \#R12.


\bibitem{RWAK07}
A. Raposo, H.J. Weber, D.E. Alvarez, and M. Kirchback, \emph{Romanovski polynomials in selected physics
problems}, Central Europ. J. Phys. \textbf{5} (2007), 253–284.



\bibitem{TW94}
C.A.~Tracy and H.~Widom, \emph{Level-spacing distributions and the Airy kernel}, Commun. Math. Phys. 159 (1994), 151–174.

\bibitem{TW94a}
C.A. Tracy and H. Widom, \emph{Level-spacing distributions and the Bessel kernel}, Commun. Math. Phys. 161 (1994),
289–309.


\bibitem{TE51}
 F. Tricomi and A. Erdélyi, \emph{The asymptotic expansion of a ratio of gamma functions}, Pacific J. Math.,
\textbf{1} (1951), 133--142


\bibitem{Ul85}
N. Ullah, \emph{Probability density function of the single eigenvalue outside the semicircle using the exact Fourier transform},
J. Math. Phys. \textbf{26}, (1985), 2350–2351.


\bibitem{VA01}
W.~Van Assche, \emph{Erratum to “Weighted zero distribution for polynomials orthogonal on an
infinite interval”}, SIAM J. Math. Anal. \textbf{32} (2001), 1169-1170.


\bibitem{WF14}
N.S. Witte and P.J. Forrester, \emph{Moments of the {G}aussian $\beta$ ensembles
  and the large {$N$} expansion of the densities}, J. Math. Phys. \textbf{55}
  (2014), 083302.


  \bibitem{YZ23}
L.~Yao and L.~Zhang, \emph{Asymptotic expansion of the hard-to-soft edge transition},  arXiv:2309.06733.


\end{thebibliography}
\end{document}